\documentclass[lettersize,journal]{IEEEtran}
\usepackage{amsmath,amsfonts,amssymb,amsthm}
\usepackage{algorithmic}
\usepackage{algorithm}
\usepackage{array}
\usepackage[h]{esvect}
\usepackage[caption=false,font=normalsize,labelfont=sf,textfont=sf]{subfig}
\usepackage{textcomp}
\usepackage{float}
\usepackage{url}
\usepackage{verbatim}
\usepackage{graphicx}
\usepackage{cite}
\usepackage{multirow}
\usepackage{xcolor}
\newcommand{\subparagraph}{}
\usepackage{titlesec}
\usepackage{makecell}

\newtheorem{lemma}{Lemma}
\newtheorem{proposition}{Proposition}

\titlespacing{\section}{0pt}{0pt}{0pt}
\titlespacing{\subsection}{0pt}{0pt}{0pt}
\titlespacing{\subsubsection}{0pt}{0pt}{0pt}

\newcommand{\bottomcaption}{%
\setlength{\abovecaptionskip}{0pt}%
\setlength{\belowcaptionskip}{0pt}%
\caption}
\newcommand{\topcaption}{%
\setlength{\abovecaptionskip}{0pt}%
\setlength{\belowcaptionskip}{0pt}%
\caption}

\allowdisplaybreaks[4]

\begin{document}

\title{Understanding the Performance of Learning Precoding Policy with GNN and CNNs}

\author{Baichuan Zhao, Jia Guo, and Chenyang Yang
\thanks{A part of this work has been published in IEEE WCNC 2022\cite{zhao2022learning}. }
}



\maketitle

\begin{abstract}
Learning-based precoding has been shown able to be implemented in real-time, jointly optimized with channel acquisition, and robust to imperfect channels. Nonetheless, existing works rarely explain when and why a deep neural network (DNN) for learning precoding policy can perform well. In this paper, we strive to understand the learning performance by taking baseband precoding as an example, for which the optimal precoding matrices of several problems such as sum rate maximization have mathematical structure. Toward this goal, we design a graph neural network (GNN) with edge-update mechanism, whose inductive bias matches to the precoding policy, and analyze its connection to the commonly used convolutional neural networks (CNNs).  Noticing that the learning performance can be decomposed into approximation and estimation errors, which depend on the smoothness of a policy and the inductive bias of a DNN, we analyze in which system settings the precoding policy is harder to be approximated by a DNN and how the inductive biases introduced by parameter sharing affect estimation errors. We proceed to derive the estimation error bounds of the DNNs. Simulations validate our analyses and verify the gain of GNN over CNNs in terms of reducing sample complexity. \end{abstract}


\vspace{2mm}\section{Introduction}

Precoding is an important technique for supporting high capacity of multi-input multi-output (MIMO) systems \cite{boccardi2014mimo}. Traditional methods for optimizing precoding resort to generic optimization tools or customized numerical algorithms such as the successive convex approximation (SCA) algorithm \cite{tran2012fast}
 and the weighted minimum mean square error (WMMSE) algorithm \cite{shi2011iteratively}. Although superior performance can be achieved, these algorithms are hard to be implemented in practice since they are unable to meet the stringent latency requirement in the fifth generation and beyond systems due to the high computational complexity \cite{zhang2019deep}.

To reduce inference complexity, deep neural networks (DNNs) have been designed to optimize precoding in a variety of multi-antenna systems \cite{hu2020iterative, kim2021learning, elbir2021family,attiah2022deep, zhang2021model,yuan2020transfer,kim2020deep, shi2021deep, chen2020sub, bo2021deep, lei2021ci, jiang2021learning,  kang2022mixed, he2020model,xia2020deep,kim2022bipartite}. Moreover, it has been shown that the learned precoding policies are robust to imperfect channels \cite{hu2020iterative,kim2021learning,elbir2021family,attiah2022deep,zhang2021model}, and the precoding policies can be optimized in an end-to-end manner without explicit channel acquisition \cite{attiah2022deep, zhang2021model, jiang2021learning, kang2022mixed}. Nonetheless, the training complexity of the DNNs for learning
precoding is high when the numbers of antennas and users are large \cite{yuan2020transfer}, which impedes their practical usage in dynamic wireless systems.

Existing learning-based methods for optimizing precoding can be classified into two categories. The first category is purely data-driven, where a DNN is designed for learning a precoding policy (i.e., a mapping from channel matrix to precoding matrix) \cite{bo2021deep, chen2020sub, attiah2022deep, elbir2021family}. These methods can be used for optimizing precoding with diverse objectives and constraints simply by selecting different loss functions for training, while the downsides are the high training complexity and weak interpretability. In the second category, mathematical models are leveraged for designing model-driven DNNs. For example, the structure of the optimal precoding matrices of several precoding problems was harnessed in \cite{kim2020deep,shi2021deep,zhang2021model,yuan2020transfer,xia2020deep,kim2022bipartite}, where DNNs were designed for learning power allocation policies, from which the precoding matrix can be recovered. Another example is deep unfolding, which was proposed to optimize precoding in \cite{hu2020iterative,kang2022mixed,he2020model}. These methods are with low training complexity due to the simplified functions for learning and are partially interpretable due to the structure of
the optimal solution or iterative algorithms. However, these model-driven DNNs are only applicable to specific problems or algorithms. For instance, only for several baseband precoding problems (say sum-rate maximization, signal-to-interference-plus-noise ratio
balancing, and minimum-rate maximization), the optimal solutions are with known structures. Yet for many other precoding problems such as energy efficient maximization precoding, hybrid analog and baseband precoding, and coordinated beamforming, the optimal solutions are without explicit structures.

To reduce the training complexity of DNNs, another approach alternative to model-driven is to introduce inductive biases, which can be a set of assumptions on which functions
are learnable by a DNN (i.e., the hypothesis space of the DNN). By imposing constraints on the function family that a DNN can represent in the form of architecture assumption \cite{battaglia2018relational}, inductive bias determines the structure of weight matrix of the DNN. For a structural DNN with appropriate
inductive bias, its learning performance with a given training set is better and its training complexity to achieve an expected performance is lower than fully connected neural networks (FNNs) \cite{baxter2000model}.

An example of structural DNN is convolutional neural network (CNN), which was originally proposed for learning image-related tasks efficiently by capturing the shift invariance and locality properties of images \cite{lecun1989backprop}. Encouraged by the great success of CNNs in many fields, most existing studies employed CNNs for learning precoding policies \cite{bo2021deep,kim2020deep,shi2021deep,chen2020sub,zhang2021model,elbir2021family}
or used CNN-based solution as a baseline \cite{hu2020iterative}, with the intuition that the channel matrix can
be regarded as an image. The CNNs used in these works are composed of linear convolutional layers cascaded with fully connected layers (referred to as F-LCNN in the sequel), which are the same as the first CNN proposed in \cite{lecun1989backprop} and are not proved to be shift invariant. The motivation of these works to use CNNs is to decrease trainable parameters \cite{shi2021deep,elbir2021family} or simply not mentioned \cite{bo2021deep, kim2020deep, chen2020sub, hu2020iterative, zhang2021model}, while only marginal performance gain was observed over FNN \cite{kim2020deep}.

Another example is graph neural network (GNN), which can leverage  the permutation equivariance (PE) property of a policy and the topology information of a graph. PE properties are widely existed in wireless policies \cite{sun2022improving,liu2023multidimensional}. By judiciously designing a GNN such that its learnable functions satisfy the same PE property of a policy, the GNN can learn the policy efficiently  \cite{eisen2020optimal,guo2021learning,liu2023multidimensional}. In \cite{jiang2021learning}, a vertex-GNN, which updates hidden representations of vertices, was designed to learn a precoding policy from the received pilots in a reflecting intelligent surface aided system. The vertex-GNN was shown generalized well to unseen number of users at low signal-to-noise ratio (SNR). However, millions of samples were used  in \cite{jiang2021learning} for training the GNN, since the equivariance of the policy to the permutation of antennas was overlooked.
In \cite{liu2023multidimensional}, multidimensional GNNs (MDGNNs) were proposed as a unified framework to learn a multitude of permutable wireless policies, which can exploit all possible permutations.

While the inductive biases have been introduced
in the aforementioned CNNs and GNNs, when and why the
DNNs for learning precoding can perform well remain elusive.

\vspace{2mm}\subsection{Related Works}
The learning performance  of DNNs can be decomposed into  estimation error and
approximation error \cite{ shalev2014understanding}.

The \textit{estimation error} is studied by statistical learning theory \cite{anthony1999neural} and computational learning theory \cite{kearns1994introduction}, which
respectively concern about the estimation error and the sample complexity that can be derived from each other.
In the literature of statistical learning theory \cite{li2018tighter}, there are two major research directions. The first direction is to establish tighter error bounds {for a given DNN architecture}, where the tightness is examined by comparing with existing results in the form of big-$\mathcal{O}$ \cite{garg2020generalization, li2018tighter, lin2019generalization, arora2018stronger}, or by analyzing the gap between the upper and lower bounds \cite{bartlett2019nearly}. The second direction is to develop error bounds {for DNNs with different architectures}, which are used to show the superiority of one DNN architecture and provide guidance for selecting DNNs. Toward this goal, these works analyze the relationship of the estimation error bounds with the hyper-parameters of DNNs, and compare the error bounds of DNNs with different architectures without providing tightness guarantees \cite{bousquet2003new}.

The \textit{approximation error} is studied by approximation theory \cite{devore2021neural}, which analyzes how well a DNN with specific architecture and scale (e.g., the depth and the width of a FNN) can approximate a target function under some assumptions on the {smoothness} of the function. In the literature of approximation theory\cite{yarotsky2017error, shen2020deep}, different metrics were used for measuring the smoothness of a function to be learned. For example, an approximation error bound of a ReLU neural network for learning Lipschitz continuous function was derived in \cite{yarotsky2017error}, which depends on a Lipschitz constant used to measure the smoothness of a target function in Lipschitz space. However, existing papers never consider which target functions satisfy the smoothness assumptions and when the target functions are smoother under these metrics.

In the field of wireless communications, only a few works have analyzed the estimation errors of DNNs by using the statistical learning theory \cite{hu2020deep, shrestha2023optimal} or computational learning theory\cite{shen2022graph}.
In \cite{hu2020deep}, an estimation error bound of a FNN used for channel acquisition was derived for showing that the bound converges to zero with infinite training samples. In \cite{shrestha2023optimal}, a special vertex-GNN was used to classify irrelevant nodes on the search tree of a joint beamforming and antenna selection algorithm, where an estimation error bound was derived for further obtaining the classification accuracy. In \cite{shen2022graph}, asymptotical estimation error bounds of FNN and a vertex-GNN were analyzed under the probably approximately correct learning framework. It was shown that the estimation error bound and the sample complexity of the GNN are respectively reduced to $\mathcal{O}(1/n)$ and $\mathcal{O}(1/n^2)$ of FNN where $n$ is the number of vertices in the graph. However, the bound of GNN is only applicable to the vertex-GNN that learns over homogeneous graph.

Although these works theoretically analyzed the learning performance of wireless policies in terms of the estimation error bounds, they did not consider the structures of the weight matrices and the resulting inductive biases of the DNNs. Therefore, they cannot explain why a DNN is more efficient for learning a policy than  another. Besides, CNNs that have been widely employed for learning wireless policies were not considered. Moreover, these works did not analyze the approximation error.  Specifically, in \cite{shrestha2023optimal}, the approximation error was simply assumed to be bounded by an unknown constant, while in \cite{hu2020deep} and \cite{shen2022graph}, it was only mentioned that FNNs with sufficient width and depth can well approximate arbitrary functions. As a consequence, these works cannot explain why a learned policy does not perform well in some settings. So far, no existing works of intelligent wireless communications have ever analyzed how the smoothness of a policy affects the learning performance.

\vspace{2mm}\subsection{Motivations and Contributions}\label{motivation-contribution}

It was mentioned in \cite{yuan2020transfer} that precoding policy is hard to learn since the input and output of a DNN for learning precoding are with high dimension due to the large numbers of antennas and users. Maybe inspired by such an observation, most existing works for optimizing precoding resorted to model-driven DNNs. For instance, only key variables such as allocated powers were learned in \cite{kim2020deep,shi2021deep,zhang2021model,yuan2020transfer,xia2020deep,kim2022bipartite} by leveraging the structure of the optimal precoding matrix. However, is it really the high dimension that incurs the hardness of learning the precoding policy?


To provide useful insight into the design of DNNs for optimizing  precoding, this paper addresses the following questions: (Q1) what is the impact of the inductive biases of DNNs on the estimation error, and (Q2) in which scenarios a precoding policy is harder to learn, i.e., the DNN for learning the policy needs more trainable parameters to reduce the approximation error?

To answer the first question, we analyze the PE property of the precoding policy, design an edge-GNN with the same PE property as the policy, which updates hidden representations of edges, and show the connection between the edge-GNN and two CNNs. One is circular CNN (CCNN) consisting of only circular convolutional layers, which was recently introduced to learn from panoramic images \cite{schubert2019circular}. The other is linear CNN (LCNN) consisting of only linear convolutional layers, which is the building block of the commonly-used F-LCNN. To answer the second question, we analyze how the SNR, channel correlation, and the numbers of antennas and users affect the smoothness of a precoding policy. Considering that precoding consists of power allocation and beamforming, we also analyze the smoothness of power allocation and beamforming policies.

We design an edge-GNN rather than adopting a vertex-GNN for the following reasons. (i) The edge-GNN has simpler architecture, since both features and actions are defined on edges in the precoding graph. If a vertex-GNN is designed, an extra readout layer is needed to transform the hidden representations of vertices to the actions of edges. (ii) The update equation of the edge-GNN can be expressed in matrix form by selecting simple aggregation and combination functions, from which we can compare the structures of weight matrices and then the inductive biases among GNN, CNNs and FNN. If a vertex-GNN is considered, complicated non-linear processing functions have to be used to avoid the performance degradation caused by the non-distinguishable channel coefficients after aggregation \cite{liu2023multidimensional}. As a consequence, the update equation of the vertex-GNN can not be expressed in matrix form. (iii) It is worth mentioning that an edge-GNN can be transformed into a vertex-GNN by using the dual hyper-graph transformation proposed in \cite{jo2021edge}. Specifically, the edges and vertices of a graph can be first transformed into the vertices and hyper-edges of a hyper-graph, and then a vertex-GNN can be used to update the hidden representations of the vertices over the hyper-graph. We do not consider this method because it induces the same update equation as the designed edge-GNN but requires an extra transformation process that is not intuitive. 

For easy elaboration, we consider the basedband precoding in multi-user multi-input single-output (MISO) system, but the analyses can be extended to other systems.

The main contributions are summarized as follows:
\begin{itemize}
\item To analyze why GNN is more efficient for learning precoding policy than CNNs and FNN, we first design an edge-GNN with matched PE property to the policy (namely EdgeGNN)  to learn the precoding policy over a heterogeneous graph. Then, we show the connection among the structures of weight matrices and the resulting inductive biases of the EdgeGNN with two CNNs, and derive their estimation error bounds by resorting to statistical learning theory.
\item To analyze when a precoding policy is harder to learn for DNNs, we analyze the impact of the system and environment parameters on the smoothness of the policy, and identify the key factors that affect the approximation error of a DNN for learning precoding policy.
\item Extensive simulations are provided to validate our analyses, and quantify the impact of the inductive biases and the key factors on the performance.
\end{itemize}

To the best of our knowledge, the conference version of this paper \cite{zhao2022learning} is the first work to design an edge-GNN for learning wireless policies. In \cite{zhao2022learning}, the two-dimensional-PE (2D-PE) property of the precoding policy was analyzed, and the edge-GNN was introduced, whose inductive bias was compared with CNNs. In this journal version, we provide the motivation and the design details of the EdgeGNN, analyze how the inductive biases affect the learning performance theoretically, and identify in which scenarios the policy is harder to learn.

Different from other works employing GNNs for learning precoding \cite{liu2023multidimensional, kim2022bipartite}, the goal of this paper is to analyze under which system settings and with what architectures the DNNs can
perform well for learning a precoding policy,
which were missing in  \cite{kim2022bipartite} and \cite{liu2023multidimensional}. Besides, in contrast to the data-driven EdgeGNN that learns the precoding policy, the bipartite GNN in \cite{kim2022bipartite} is a model-driven vertex-GNN that learns the power allocation policy for recovering the precoding matrix. Different from the MDGNNs proposed in \cite{liu2023multidimensional} that satisfy the PE property of wireless policies by introducing parameter sharing into FNNs, the EdgeGNN satisfies the PE property by updating the hidden representations over a graph, which is not a special case of the MDGNNs.

The rest of the paper is organized as follows. Sections II and III introduce the EdgeGNN and two CNNs to learn the precoding policy, respectively. In Section IV, we analyze how the system and environment parameters and the inductive biases of DNNs affect the learning performance of precoding. Simulation results and conclusions are given in Section V and Section VI, respectively.

Notations: $\mathbf{0}$ and $\mathbf{1}$ denote all-zero and all-one matrices, respectively. $\mathbf{I}$ denotes identity matrix. $\mathrm{vec}(\mathbf{A})$ denotes the operation of vectorizing $\mathbf{A}$ into a vector.
$(\cdot)^\mathsf{T}$ and $(\cdot)^\mathsf{H}$ denote the transpose and the conjugate transpose, respectively.
$|\cdot|$ denotes magnitude. $||\cdot||$ denotes
the Frobenius norm of matrix or two norm of vector. $||\cdot||_\sigma$ denotes the spectral norm of matrix. $\otimes$, $\circledast$ and $\odot$ denote Kronecker product, circular convolution and Hadamard product, respectively.

\vspace{2mm}\section{EdgeGNN for Learning Precoding Policy}

In this section, we first briefly introduce a precoding policy and its 2D-PE
property. Then, we construct a graph, and design the EdgeGNN for learning the policy over the graph.

\vspace{2mm}\subsection{Precoding Policy  and PE Property}

Consider a downlink system, where a base station (BS) equipped with $N$ antennas transmits to $K$ users each with a single antenna. Denote the channel vector of the $k$th user as $\mathbf{h}_k=[h_{k,1}\cdots,h_{k,N}]^\mathsf{T}\in\mathbb{C}^{N\times1}$, where $h_{k,n}$ is the channel coefficient between the $k$th user and the $n$th antenna. After obtaining the channel matrix $\mathbf{H}=[\mathbf{h}_1,\cdots,\mathbf{h}_K]^\mathsf{H}\in\mathbb{C}^{K{\times}N}$, a precoding matrix $\mathbf{V}=[\mathbf{v}_1,\cdots,\mathbf{v}_K]^\mathsf{H}\in\mathbb{C}^{K{\times}N}$ is optimized at the BS to maximize a utility function $U(\mathbf{H},\mathbf{V})$ under a set of constraints $\mathbf{c}(\mathbf{H},\mathbf{V})$, i.e.,
\begin{equation}
\label{eq_opt}
\max_{\mathbf{V}}~~U(\mathbf{H},\mathbf{V})  ~~
s.t. ~~ \mathbf{c}(\mathbf{H},\mathbf{V}){\leq}0,
\end{equation}
where $\mathbf{v}_k=[v_{k,1},\cdots,v_{k,N}]^\mathsf{H}\in\mathbb{C}^{K{\times}N}$ is the precoding vector to the $k$th user.

Denote the precoding policy as $\mathbf{V}^\star=f^\star(\mathbf{H})$, where $\mathbf{V}^\star$ is the optimized solution of the problem in \eqref{eq_opt} with given $\mathbf{H}$. Assume that the utility
and constraint functions do not change with the orders of users and antennas, which holds
for commonly used utility functions (e.g., spectral and energy efficiencies) and constraints (e.g., power and data rate constraints) \cite{eisen2020optimal}. Then, the precoding policy exhibits the 2D-PE property: $\boldsymbol{\Pi}_1^\mathsf{T}\mathbf{V}^\star\boldsymbol{\Pi}_2=f^\star(\boldsymbol{\Pi}_1^\mathsf{T}\mathbf{H}\boldsymbol{\Pi}_2)$, where $\boldsymbol{\Pi}_1$ and $\boldsymbol{\Pi}_2$ are arbitrary permutation matrices \cite{zhao2022learning}.

\vspace{2mm}\subsection{Graph for Precoding}
Denote a graph as $\mathcal{G}=(\mathcal{V},\mathcal{E})$, where $\mathcal{V}$ is a set of vertices and $\mathcal{E}$ is a set of edges. Each vertex and each edge may be associated with a \textit{feature} or an \textit{action}.

Constructing appropriate graphs is the premise of applying GNNs, while the graph is not unique to optimize a policy. To harness the 2D-PE property, a \textit{precoding graph} for learning the precoding policy is constructed as follows.

The precoding graph is a heterogeneous graph with two types of vertices, as illustrated in Fig.~\ref{fig_relation}. Each antenna is a vertex and each user is a vertex of another type. The link between them is an edge and all edges belong to the same type.
There is no feature or action on vertices. The feature of the edge between the $k$th user vertex and
the $n$th antenna vertex (denoted as edge $(k,n)$) is  $h_{k,n}$, and
the features of all the edges can be expressed as $\mathbf{H}$. The action of edge $(k,n)$ is the $n$th element of the precoding vector for the $k$th user (i.e., $v_{k,n}$), and all the actions can be expressed as $\mathbf{V}$.

\begin{figure}[t]
\centerline{\includegraphics[width=0.7\linewidth]{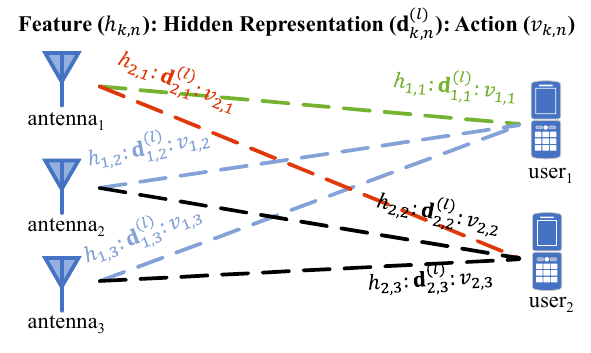}}\vspace{-2mm}
\bottomcaption{Precoding graph and hidden representation, $N=3$ and $K=2$.}
\label{fig_relation}
\end{figure}

\vspace{2mm}\subsection{An Edge-GNN Learning Over the Precoding Graph}

To learn the mapping from the features to the actions that are all defined on edges, a natural
choice is to update the hidden representations of edges in each layer, where the information from the neighboring edges of each edge is first aggregated at its neighboring vertices and then combined. Then, the representations of the edges in the last layer are the actions.
In the precoding graph, the neighboring vertices of edge $(k,n)$ are the $k$th user vertex and the $n$th antenna vertex, while the neighboring edges of edge $(k,n)$ are the edges connected with the $k$th user vertex or the $n$th antenna vertex. Further considering that all the vertices are without features, the hidden representation of edge $(k,n)$ in the $l$th layer of the designed edge-GNN, i.e., the EdgeGNN, is updated as follows.
\begin{itemize}
\item[(1)]\textbf{Aggregation:} The aggregated outputs of the $k$th user vertex and the $n$th antenna vertex are respectively
\begin{IEEEeqnarray*}{c}\IEEEyesnumber\label{eq_gnn_aggr}
\mathbf{u}_k^{(l)}=\mathrm{PL}^{(l)}_1\left(\left\{q_1^{(l)}(\mathbf{d}^{(l-1)}_{k,n}):n=1,\cdots,N\right\}\right),\IEEEyessubnumber\\
\mathbf{a}_n^{(l)}=\mathrm{PL}^{(l)}_2\left(\left\{q_2^{(l)}(\mathbf{d}^{(l-1)}_{k,n}):k=1,\cdots,K\right\}\right),\IEEEyessubnumber
\end{IEEEeqnarray*}
\noindent where $q_1^{(l)}(\cdot),q_2^{(l)}(\cdot)$, $\mathrm{PL}_1^{(l)}(\cdot),\mathrm{PL}_2^{(l)}(\cdot)$ are the processing and pooling functions in the $l$th layer for user vertices and antenna vertices, respectively. $\mathrm{PL}_1^{(l)}(\cdot)$ and $\mathrm{PL}_2^{(l)}(\cdot)$ are invariant to the permutation of their input variables.
\item[(2)]\textbf{Combination:} The hidden representation of edge $(k,n)$ in the $l$th layer can be obtained as,
\begin{equation}\label{eq_gnn_comb}
\mathbf{d}_{k,n}^{(l)}=\mathrm{CB}^{(l)}\left(\mathbf{d}_{k,n}^{(l-1)},\mathbf{u}_k^{(l)},\mathbf{a}_l^{(l)}\right),
\end{equation}
\noindent where $\mathbf{d}^{(l)}_{k,n}\in\mathbb{R}^{C_l{\times}1}$, $C_l$ is the number of hidden features of the edge, and $\mathrm{CB}^{(l)}(\cdot)$ is the combination function.
\end{itemize}

The hidden representations of all edges in the $l$th layer form a tensor $\mathbf{D}^{(l)}=[\mathbf{d}^{(l)}_{k,n}]\in\mathbb{R}^{K\times N \times C_l}$. The input of the EdgeGNN is $\mathbf{D}^{(0)}\in\mathbb{R}^{K\times N \times 2}$, which is composed of the real and imaginary parts of $\mathbf{H}$. Such an EdgeGNN can be used for learning precoding policy with different utility functions
and constraints by simply training with different loss functions.

\vspace{2mm}\subsection{Structure of Weight Matrix of the EdgeGNN}
A variety of processing, pooling, and combination functions can be chosen. To express the update equation of the EdgeGNN in matrix form for the sake of connecting to other DNNs later, we employ linear processing, $\mathrm{sum}(\cdot)$ pooling, and a simple combination function. Then, the hidden representations of edge $(k,n)$ in the $l$th layer are updated as \vspace{-2mm}
\begin{IEEEeqnarray*}{lc}\IEEEyesnumber
\textbf{Aggregation: } & \IEEEnonumber\\
\IEEEeqnarraymulticol{2}{c}{\mathbf{u}_k^{(l)}=\sum_{n=1}^N \mathbf{P}_\mathrm{G} \mathbf{d}_{k,n}^{(l-1)},~~\mathbf{a}_n^{(l)}=\sum_{k=1}^K \mathbf{Q}_\mathrm{G} \mathbf{d}_{k,n}^{(l-1)},\IEEEyessubnumber\label{eq_aggregator}}\\
\textbf{Combination: } & \IEEEnonumber\\ \mathbf{d}^{(l)}_{k,n}=\phi\left(\mathbf{O}_\mathrm{G}\mathbf{d}^{(l-1)}_{k,n}+\mathbf{u}^{(l)}_k+\mathbf{a}_n^{(l)}\right),\IEEEyessubnumber\label{eq_combiner}
\end{IEEEeqnarray*}
\noindent where $\phi(\cdot)$ is an activation function, and $\mathbf{O}_\mathrm{G},\mathbf{P}_\mathrm{G},\mathbf{Q}_\mathrm{G}\in\mathbb{R}^{C_l{\times}C_{l-1}}$ are the trainable weight matrices. Here and in the sequel, we omit the superscript $(l)$ of the activation function and the
weight matrices for notational simplicity.

By vectorizing the tensor $\mathbf{D}^{(l)}$ as $\mathrm{vec}(\mathbf{D}^{(l)})\triangleq \left[\mathbf{d}^{(l)\mathsf{T}}_{1,1},\cdots,\mathbf{d}^{(l)\mathsf{T}}_{K,1},\cdots,\mathbf{d}^{(l)\mathsf{T}}_{1,N},\cdots,\mathbf{d}^{(l)\mathsf{T}}_{K,N}\right]^\mathsf{T}\in\mathbb{R}^{KNC_l\times1}$ and substituting \eqref{eq_aggregator} into \eqref{eq_combiner}, the update equation for the hidden representations of all edges in the $l$th layer can be expressed in matrix form as
\begin{equation}\label{eq_forward}
\mathrm{vec}(\mathbf{D}^{(l)})=\phi\left(\mathbf{W}_\mathrm{G}\mathrm{vec}(\mathbf{D}^{(l-1)})\right),
\end{equation}
\noindent where the trainable weight matrix $\mathbf{W}_\mathrm{G}$ is composed of two sub-matrices $\mathbf{B}_\mathrm{G}$ and $\mathbf{C}_\mathrm{G}$ that are further composed of $\mathbf{P}_\mathrm{G}$, $\mathbf{Q}_\mathrm{G}$ and $\mathbf{O}_\mathrm{G}^\prime=\mathbf{O}_\mathrm{G}+\mathbf{P}_\mathrm{G}+\mathbf{Q}_\mathrm{G}$ with the following structures
\setlength{\arraycolsep}{2pt}
\begin{IEEEeqnarray*}{c}
\mathbf{W}_\mathrm{G}
=\left[\begin{array}{cccc}
\mathbf{B}_\mathrm{G} & \mathbf{C}_\mathrm{G} & \cdots & \mathbf{C}_\mathrm{G} \\
\mathbf{C}_\mathrm{G} & \mathbf{B}_\mathrm{G} & \cdots & \mathbf{C}_\mathrm{G} \\
\vdots & \vdots & \ddots & \vdots \\
\mathbf{C}_\mathrm{G} & \mathbf{C}_\mathrm{G} & \cdots & \mathbf{B}_\mathrm{G} \\
\end{array}\right],\\
\mathbf{B}_\mathrm{G}=\left[\begin{array}{cccc}
\mathbf{O}_\mathrm{G}^\prime & \mathbf{Q}_\mathrm{G} & \cdots & \mathbf{Q}_\mathrm{G} \\
\mathbf{Q}_\mathrm{G} & \mathbf{O}_\mathrm{G}^\prime & \cdots & \mathbf{Q}_\mathrm{G} \\
\vdots & \vdots & \ddots & \vdots \\
\mathbf{Q}_\mathrm{G} & \mathbf{Q}_\mathrm{G} & \cdots & \mathbf{O}_\mathrm{G}^\prime \\
\end{array}\right],
\mathbf{C}_\mathrm{G}=\left[\begin{array}{cccc}
\mathbf{P}_\mathrm{G} & \mathbf{0} & \cdots & \mathbf{0} \\
\mathbf{0} & \mathbf{P}_\mathrm{G} & \cdots & \mathbf{0} \\
\vdots & \vdots & \ddots & \vdots \\
\mathbf{0} & \mathbf{0} & \cdots & \mathbf{P}_\mathrm{G} \\
\end{array}\right].\IEEEyesnumber\label{eq_structured_parameter}
\end{IEEEeqnarray*}
\setlength{\arraycolsep}{5pt}

To understand how the weight matrix is obtained by updating the hidden representations over the graph, we take edge $(1,1)$ in Fig.~\ref{fig_relation} as an example. Its hidden representation (i.e., $\mathbf{d}_{1,1}^{(l)}$ in green color) is obtained from the previous hidden representations of: (i) the same edge (i.e., $\mathbf{d}_{1,1}^{(l-1)}$ in green), (ii) the edges linked to the first user vertex (i.e., $\mathbf{d}_{1,n}^{(l-1)},{\forall}n{\neq}1$ in blue), (iii) the edges connected to the first antenna vertex (i.e., $\mathbf{d}_{k,1}^{(l-1)},{\forall}k{\neq}1$  in red), and (iv) other edges that are not neighboured to edge $(1,1)$ (i.e., $\mathbf{d}_{k,n}^{(l-1)},{\forall}k\neq1,n\neq1$ in gray). These four terms are respectively multiplied by the weight matrices $\mathbf{O}_\mathrm{G}^\prime$, $\mathbf{Q}_\mathrm{G}$, $\mathbf{P}_\mathrm{G}$ and $\mathbf{0}$.

\vspace{2mm}\section{CNNs for Learning Precoding Policy}
In this section, we first introduce two CNNs for learning the precoding policy, which have never been used for learning wireless policies. To show the connection of the two CNNs to the EdgeGNN, we then express the input-output relations of the two CNNs in matrix form.

\vspace{2mm}\subsection{Structure of Weight Matrices of Two CNNs}

F-LCNN has been used for learning precoding policies
\cite{bo2021deep, kim2020deep, shi2021deep, chen2020sub, zhang2021model, elbir2021family}, where the input and output tensors are respectively composed of the real and imaginary parts of $\mathbf{H}$ and $\mathbf{V}$, denoted as $\mathbf{D}^{(0)}\in\mathbb{R}^{K\times N \times C_0}$ and $\mathbf{D}^{(L)}\in\mathbb{R}^{K\times N \times C_L}$, where $C_0=C_L=2$. Since the input and output are with the same size, pooling layers were not used, and the size of the input tensor needs to be padded to avoid the size shrinking of the output tensor after convolution. Single stride was considered in these works, a possible reason is that more zeros need to be padded with a larger stride.

Since the inductive bias of the F-LCNN comes from the structural convolutional layers, we consider the LCNN that only consists of linear convolutional layers without fully connected layers. To analyze the connection to the inductive bias of the EdgeGNN, we also consider CCNN that only consists of circular convolutional layers. Both LCNN and CCNN are with single stride and without pooling layers.

\subsubsection{Circular Convolutional Layer}

CCNN is with circular padding. In the $l$th layer, $\mathbf{D}^{(l-1)}$ is first convolved by $C_l$ convolution kernels denoted as  $\mathbf{W}_\mathrm{ker}^{<i>}=[\mathbf{w}^{<i>}_{k,n}]\in\mathbb{R}^{H_\mathrm{conv}{\times}W_\mathrm{conv}{\times}C_{l-1}},i=1,\cdots,C_l$, where $\mathbf{w}_{k,n}^{<i>}\in\mathbb{R}^{C_{l-1}{\times}1}$ is at the $k$th row and the $n$th column of $\mathbf{W}_\mathrm{ker}^{<i>}$, $k=1,\cdots,H_\mathrm{conv}$, $n=1,\cdots,W_\mathrm{conv}$, $H_\mathrm{conv}$ and $W_\mathrm{conv}$ are respectively the height and width of the convolution kernel. Then, $C_l$ matrices each with size of $K{\times}N$ can be obtained, and $\mathbf{D}^{(l)}$ is obtained by stacking these matrices. This process is illustrated in Fig.~\ref{fig_circ_conv}(a), where $C_l=C_{l-1}=1$, and the right bottom of $\mathbf{D}^{(l-1)}$ is padded by the left top part of the matrix  \cite{schubert2019circular}. After padding, the size of $\mathbf{D}^{(l-1)}$ becomes $(K+H_\mathrm{conv}-1)\times(N+W_\mathrm{conv}-1)$. After convolving with the kernel of size $H_\mathrm{conv}{\times}W_\mathrm{conv}$, the size of the output matrix is $K{\times}N$.

\begin{figure}[htb]
\centerline{\includegraphics[width=0.9\linewidth]{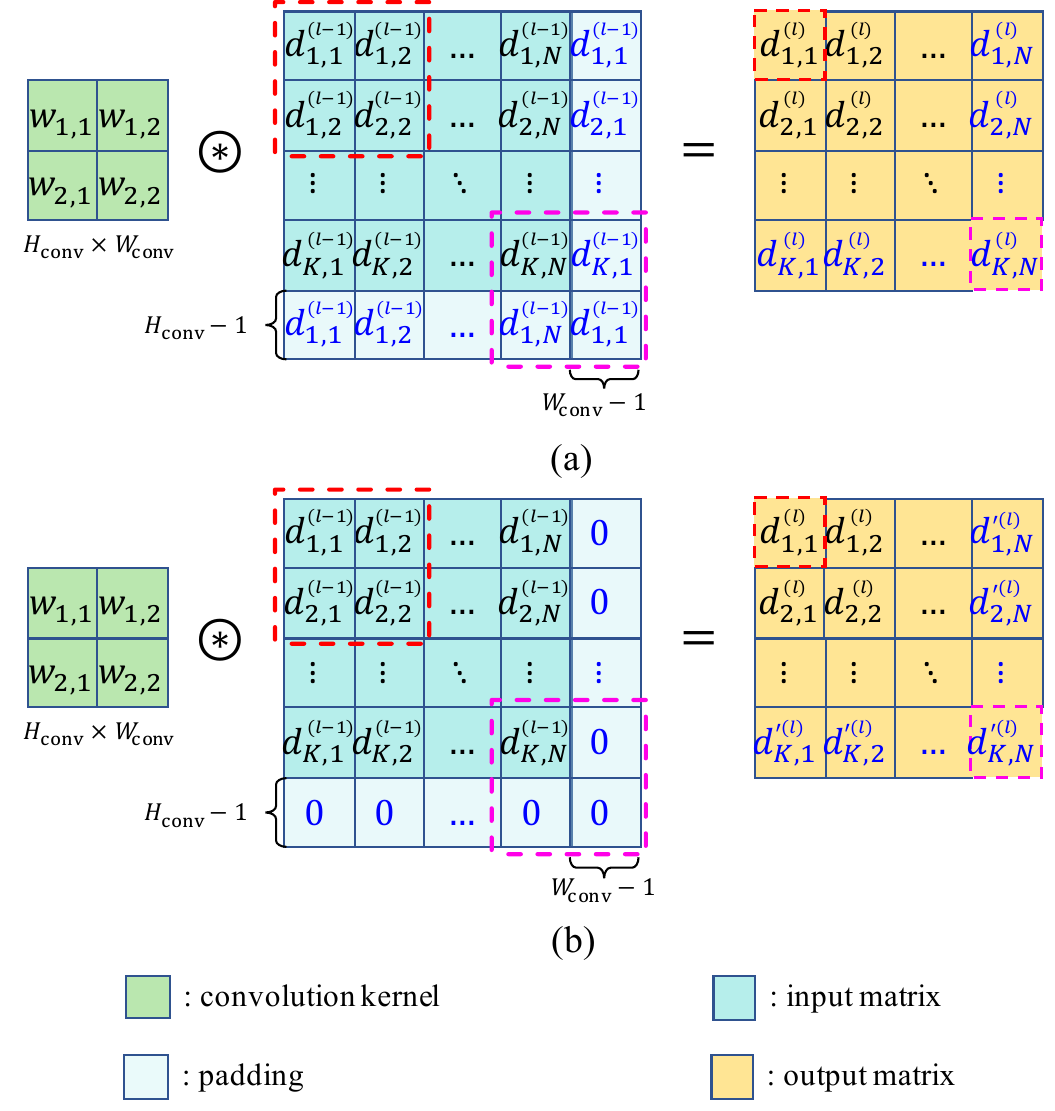}}
\bottomcaption{(a) Circular convolution and (b) linear convolution, both are with single stride and $\mathbf{W}_\mathrm{ker}$ is of size $H_\mathrm{conv}=W_\mathrm{conv}=2$. Since we take $C_{l}=C_{l-1}=1$ as example, the superscript of $\mathbf{W}_\mathrm{ker}^{<i>}=[\mathbf{w}^{<i>}_{k,n}]$ is omitted, meanwhile $\mathbf{d}_{k,n}^{(l-1)},\mathbf{d}_{k,n}^{(l)}$, and $\mathbf{w}_{k,n}$ degenerate into scalars.}
\label{fig_circ_conv}
\end{figure}

\begin{figure}[htb]
\centerline{\includegraphics[width=0.7\linewidth]{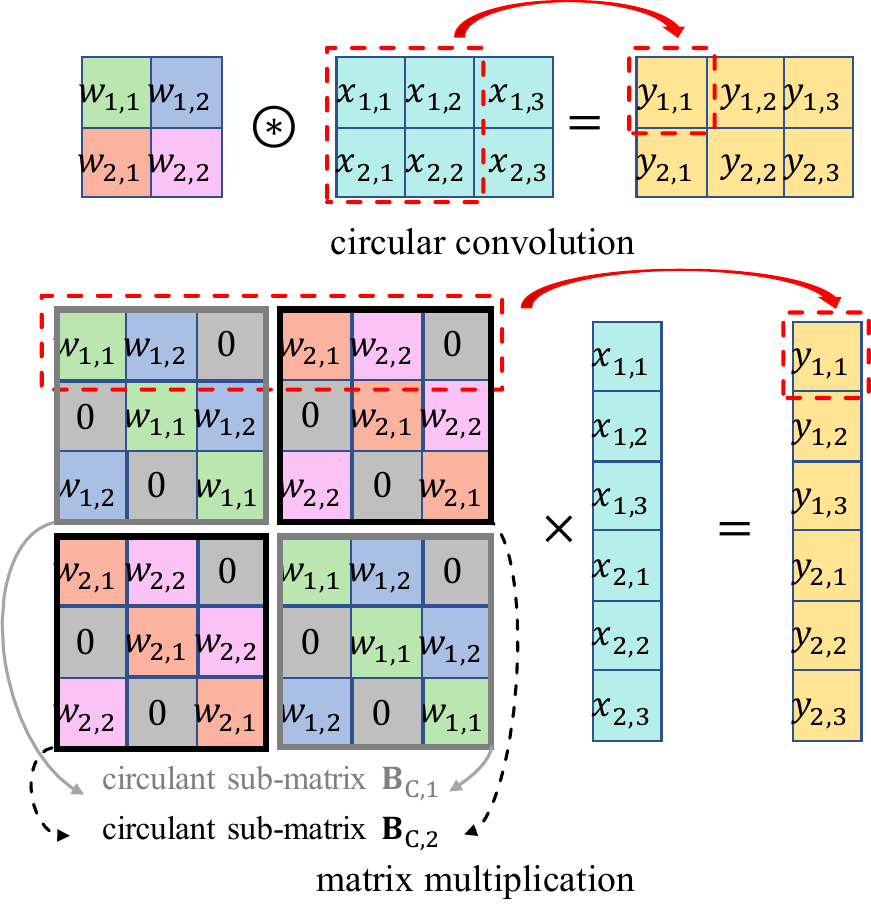}}\vspace{-2mm}
\bottomcaption{Expressing circular convolution into matrix multiplication, $K=2, N=3$, $H_\mathrm{conv}=W_\mathrm{conv}=2$, and $C_{l}=C_{l-1}=1$. $\mathbf{D}^{(l-1)}\in \mathbb{R}^{2\times 3}$ is vectorized, and the convolution kernel is re-expressed as a block circulant matrix with $2\times2$ sub-matrices, each sub-matrix is a circulant matrix of size $3\times3$. Since the size of the kernel is smaller than the size of $\mathbf{D}^{(l-1)}$, zeros are padded into each sub-matrix. Then, the same output can be obtained by circular convolution and matrix multiplication.}
\label{fig_vec_circ_conv}\vspace{2mm}
\end{figure}

Next, we express the circular convolution operation as matrix multiplication, as illustrated in Fig. \ref{fig_vec_circ_conv}.  Then, the relation between the hidden outputs of the $l$th and the $(l-1)$th layers (can also be regarded as an update equation) can be expressed in matrix form as
\begin{equation}\label{eq_forward_circ_cnn}
\mathrm{vec}(\mathbf{D}^{(l)})=\phi\left(\mathbf{W}_\mathrm{C}\mathrm{vec}(\mathbf{D}^{(l-1)})\right).
\end{equation}
The weight matrix and sub-matrices within it are respectively
\begin{IEEEeqnarray*}{c}
\label{eq_circ_structured_parameter}
\setlength{\arraycolsep}{2pt}
\mathbf{W}_\mathrm{C}
=\left[\begin{array}{cccc}
\mathbf{B}_{\mathrm{C},1} & \mathbf{B}_{\mathrm{C},2} & \cdots & \mathbf{B}_{\mathrm{C},N} \\
\mathbf{B}_{\mathrm{C},N} & \mathbf{B}_{\mathrm{C},1} & \cdots & \mathbf{B}_{\mathrm{C},N-1} \\
\vdots & \vdots & \ddots & \vdots \\
\mathbf{B}_{\mathrm{C},2} & \mathbf{B}_{\mathrm{C},3} & \cdots & \mathbf{B}_{\mathrm{C},1} \\
\end{array}\right],\\
\setlength{\arraycolsep}{2pt}
\mathbf{B}_{\mathrm{C},n}=\left[\begin{array}{cccc}
\mathbf{Q}_{n,1} & \mathbf{Q}_{n,2} & \cdots & \mathbf{Q}_{n,K} \\
\mathbf{Q}_{n,K} & \mathbf{Q}_{n,1} & \cdots & \mathbf{Q}_{n,K-1} \\
\vdots & \vdots & \ddots & \vdots \\
\mathbf{Q}_{n,2} & \mathbf{Q}_{n,3} & \cdots & \mathbf{Q}_{n,1} \\
\end{array}\right],\IEEEyesnumber
\setlength{\arraycolsep}{5pt}
\end{IEEEeqnarray*}
\noindent where $\mathbf{Q}_{n,k}=[\mathbf{w}^{<1>}_{k,n},\cdots,\mathbf{w}^{<C_l>}_{k,n}]^\mathsf{T}\in\mathbb{R}^{C_l{\times}C_{l-1}}$ are trainable weight  matrices when $n \leq W_\mathrm{conv}$, $k\leq H_\mathrm{conv}$, and $\mathbf{Q}_{n,k}=\mathbf{0}$ for other values of $k$ and $n$.

\subsubsection{Linear Convolutional Layer}\label{subsubsec_linear_conv}
The only difference between LCNN and CCNN lies in what is padded into the input matrix, as illustrated in Fig.~\ref{fig_circ_conv}(b). When  LCNN is used to learn the precoding policy, the relation between $\mathbf{D}^{(l)}$ and $\mathbf{D}^{(l-1)}$  can be expressed as
\begin{equation}
\mathrm{vec}(\mathbf{D}^{(l)})=\phi\left(\mathbf{W}_\mathrm{L}\mathrm{vec}(\mathbf{D}^{(l-1)})\right).
\end{equation}
\noindent The weight matrix and sub-matrices within it are respectively
\begin{IEEEeqnarray*}{c}\label{eq_linear_structured_parameter}
\setlength{\arraycolsep}{2pt}
\mathbf{W}_\mathrm{L}
=\left[\begin{array}{cccc}
\mathbf{B}_{\mathrm{L},1} & \mathbf{B}_{\mathrm{L},2} & \cdots & \mathbf{B}_{\mathrm{L},N} \\
\mathbf{0} & \mathbf{B}_{\mathrm{L},1} & \cdots & \mathbf{B}_{\mathrm{L},N-1} \\
\vdots & \vdots & \ddots & \vdots \\
\mathbf{0} & \mathbf{0} & \cdots & \mathbf{B}_{\mathrm{L},1} \\
\end{array}\right],\\
\setlength{\arraycolsep}{2pt}
\mathbf{B}_{\mathrm{L},n}=\left[\begin{array}{cccc}
\mathbf{Q}_{n,1} & \mathbf{Q}_{n,2} & \cdots & \mathbf{Q}_{n,K} \\
\mathbf{0} & \mathbf{Q}_{n,1} & \cdots & \mathbf{Q}_{n,K-1} \\
\vdots & \vdots & \ddots & \vdots \\
\mathbf{0} & \mathbf{0} & \cdots & \mathbf{Q}_{n,1} \\
\end{array}\right], \IEEEyesnumber\\
\setlength{\arraycolsep}{5pt}
\end{IEEEeqnarray*}
\noindent where $\mathbf{Q}_{n,k}\in\mathbb{R}^{C_l{\times}C_{l-1}}$ is the same as in \eqref{eq_circ_structured_parameter}.

\vspace{2mm}\subsection{Connection among the EdgeGNN and CNNs}\label{subsec_relation}
It is shown from \eqref{eq_structured_parameter} and \eqref{eq_circ_structured_parameter} that the CCNN degenerates into the EdgeGNN if $\mathbf{B}_{\mathrm{C},2}=\cdots=\mathbf{B}_{\mathrm{C},N}$, $\mathbf{Q}_{1,2}=\cdots=\mathbf{Q}_{1,K}$ in $\mathbf{B}_{\mathrm{C},1}$, and $\mathbf{Q}_{n,2}=\cdots=\mathbf{Q}_{n,K}=\mathbf{0}$ in $\mathbf{B}_{\mathrm{C},n}$ for $n\neq1$.

It is shown from \eqref{eq_circ_structured_parameter} and \eqref{eq_linear_structured_parameter} that linear convolutional layer is equivalent to circular convolutional layer if and only if $\mathbf{B}_{\mathrm{L},n}=\mathbf{B}_{\mathrm{C},n}=\mathbf{0}$, $n\geq2$, and $\mathbf{Q}_{1,k}=\mathbf{0}$, $k \geq 2$ (i.e., $H_\mathrm{conv}{\times}W_\mathrm{conv}=1\times1$). Otherwise, the output matrices of the linear and circular convolutional layers are different at the border  (referred to as \emph{border effect}), as illustrated in Fig.~\ref{fig_circ_conv}.

\vspace{2mm}\section{Understanding the Learning Performance}\label{sec_learning_perf}

In this section, we analyze how the environment and system parameters as well as the inductive biases affect the approximation error and estimation error of learning precoding policy with DNNs.

According to approximation theory \cite{devore2021neural}, the approximation error of a DNN depends on the smoothness of a target function to be learned (i.e., a policy) as well as the expressive capability of the DNN. According to statistical learning theory, the estimation error of a DNN depends on its hypothesis space as well as the number of training samples \cite{anthony1999neural}.

To analyze in which scenarios a precoding policy is non-smooth, which is harder to be learned by any DNN, we consider a widely-investigated precoding problem whose solution is with a well-known structure. Based on the update equations in matrix form, we derive the inductive biases and the estimation error bounds of the EdgeGNN and CNNs.\vspace{-0.2mm}


\vspace{2mm}\subsection{Preliminary}
The learning performance of a DNN, which is referred to as excess risk in \cite{shalev2014understanding}, is composed of approximation and estimation errors. For the precoding problem at hand, the learning performance of the DNN can be expressed as
\begin{IEEEeqnarray}{c}
\underbrace{\mathcal{R}(f^\star)-\mathcal{R}(\hat{f})}_{\textrm{excess risk}}=\underbrace{[\mathcal{R}(f^\star)-\mathcal{R}(\hat{f}^\star)]}_{\textrm{approximation error}}+\underbrace{[\mathcal{R}(\hat{f}^\star)-\mathcal{R}(\hat{f})]}_{\textrm{estimation error}},\IEEEeqnarraynumspace
\label{eq_decom_gen_err}
\end{IEEEeqnarray}
\noindent where $\mathcal{R}(f)=\mathbb{E}\left[U(\mathbf{H},f(\mathbf{H}))\right]$ is the average utility of a precoding policy $f(\cdot)$ over $\mathbf{H}$, $f^\star(\cdot)$ is the optimal precoding policy, $\hat{f}(\cdot)$ is a learned precoding policy, and $\hat{f}^\star(\cdot)$ is the precoding policy in the hypothesis space achieving the highest average utility. Hence, $\mathcal{R}(\hat{f}^\star)$ is the upper bound of the learning performance achieved by the DNN.

To help understand, we represent the performance gap of two polices in \eqref{eq_decom_gen_err} by the ``distance'' between them in Fig. \ref{fig_est_appr}. As illustrated in Fig. \ref{fig_est_appr}(a), the \textit{approximation error} is the distance between the target function $f^\star(\cdot)$ to be learned and the function closest to the policy representable by a given DNN, i.e., $\hat{f}^\star(\cdot)$.
The \textit{estimation error} is the distance between the best-fitted policy $\hat{f}^\star(\cdot)$ by a given DNN and the policy $\hat{f}(\cdot)$ learned by the DNN with a given set of training samples.

\begin{figure}[t]
\centerline{\includegraphics[width=1\linewidth]{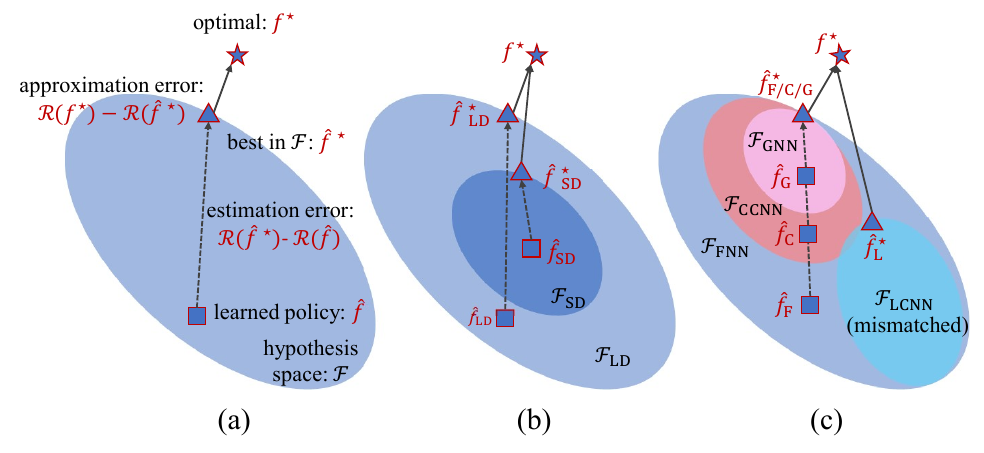}}\vspace{-2mm}
\bottomcaption{(a) Approximation error and estimation error of a DNN. (b) Impact of enlarging the scale of the DNN, ``LD'' and ``SD'' denote the DNNs with large and small scales, respectively. (c) Impact of inductive bias.}
\label{fig_est_appr}\vspace{2mm}
\end{figure}

When analyzing the approximation error of a DNN for learning a function with approximation theory \cite{devore2021neural}, the target function is always assumed satisfying specific smoothness properties such as belonging to the Lipschitz space \cite{yarotsky2017error, shen2020deep}, where the smoothness of a function $f(\cdot)$ is measured by its \textit{Lipschitz constant} defined as $L_f\triangleq\sup_{\mathbf{x}_1,\mathbf{x}_2}\frac{||f(\mathbf{x}_1)-f(\mathbf{x}_2)||}{||\mathbf{x}_1-\mathbf{x}_2||}$.

For a given DNN, the approximation error is smaller when learning a smoother target function, as proved in \cite{yarotsky2017error}. On the other hand, for a given target function, the approximation error can be reduced by improving the expressive capability of the DNN \cite{devore2021neural} through, e.g., enlarging the scale of the DNN (i.e., using wider and deeper DNN).

For a given DNN, the estimation error can be reduced by enlarging the training set. On the other hand, for a given number of training samples, one can reduce the estimation error by reducing the hypothesis space of a DNN, and a straightforward way is to reduce the width and depth of the DNN with a specific architecture.

This indicates that there is a trade-off between the approximation error and the estimation error when adjusting the scale of a DNN. Specifically, if a DNN with larger scale is used for learning a policy with a training set, its expressive capability is stronger and thus the approximation error is smaller, which however enlarges the hypothesis space and leads to a higher estimation error, as illustrated in Fig. \ref{fig_est_appr}(b).

To deal with this issue,  one can impose policy-matched inductive biases to the architecture of a DNN, which also affects its hypothesis space, in addition to the scale. Then, a better trade-off can be achieved, as shown in Fig.~\ref{fig_est_appr}(c) and as to be further explained in Section \ref{subsec_rib}.


\subsection{When Precoding Policy is Harder to Learn?}\label{subsec_appr}

Learning a less smooth policy with a DNN to achieve satisfactory performance is more challenging, regardless of the architecture of the DNN. This is because larger scale is required to reduce the approximation error of a DNN with given architecture for learning the policy, which increases the estimation error if the training samples are insufficient.

To analyze when a wireless policy is harder to learn, we analyze the smoothness of the policy under different system and environment parameters. To help understand, we take the following baseband precoding problem as an example, whose optimal precoding matrix has a well-known structure. In fact, the optimal solutions of several precoding problems with other objective functions and constraints have similar structures \cite{bjornson2014optimal}.

We consider a sum rate maximization precoding problem under the constraint of total transmit power $P_\mathrm{t}$,
\begin{IEEEeqnarray}{rcl}
\label{eq_sr_constraint}
\max_{\mathbf{V}} & \quad & \sum_{k=1}^K\mathrm{log}_2\left(1+\frac{|\mathbf{h}_k^\mathsf{H}\mathbf{v}_k|^2}{\sigma^2+\sum_{m{\neq}k}|\mathbf{h}_k^\mathsf{H}\mathbf{v}_m|^2}\right)\IEEEnonumber\\
s.t.& \quad & \sum_{k=1}^K||\mathbf{v}_k||^2{\leq}P_\mathrm{t},
\end{IEEEeqnarray}
whose optimal solution is with the following well-known structure \cite{bjornson2014optimal},
\begin{IEEEeqnarray}{rcl}
\mathbf{V}^\star&=&\left(\mathbf{I}_N+\frac{1}{\sigma^2}\mathbf{H}^\mathsf{H}\boldsymbol{\Lambda}\mathbf{H}\right)^{-1}\mathbf{H}^\mathsf{H}\mathbf{P}^\frac{1}{2}\IEEEnonumber\\
&=&\mathbf{H}^\mathsf{H}\left(\sigma^2\mathbf{I}_K+\boldsymbol{\Lambda}\mathbf{H}\mathbf{H}^\mathsf{H}\right)^{-1}\widetilde{\mathbf{P}}^\frac{1}{2},
\label{eq_duality}
\end{IEEEeqnarray}
\noindent where $\mathbf{P}=\mathrm{diag}(p_1/||(\mathbf{I}_N+\frac{1}{\sigma^2}\mathbf{H}\boldsymbol{\Lambda}\mathbf{H}^\mathsf{H})^{-1}\mathbf{h}_1||^2,\cdots,$ $p_K/||(\mathbf{I}_N+\frac{1}{\sigma^2}\mathbf{H}\boldsymbol{\Lambda}\mathbf{H}^\mathsf{H})^{-1}\mathbf{h}_K||^2)$, $\widetilde{\mathbf{P}}=\sigma^4\mathbf{P}$, $\boldsymbol{\Lambda}=\mathrm{diag}(\lambda_1,\cdots,\lambda_K)$, $p_1,\cdots,p_k$ are the allocated powers, $\sigma^2$ is the noise variance, and $\lambda_1,\cdots,\lambda_k$ are positive parameters satisfying $\sum_{k=1}^K\lambda_k=\sum_{k=1}^K{p_k}=P_\mathrm{t}$.


As can be seen in \eqref{eq_duality}, the precoding policy can be expressed as $\mathbf{V}^\star=f_\mathrm{BF}(f_\mathrm{PA}(\mathbf{H}),\mathbf{H})$, where $f_\mathrm{BF}(\cdot)$ and $f_\mathrm{PA}(\cdot)$ denote the beamforming and power allocation policies, respectively. The beamforming policy involves a matrix inverse function $f_\mathrm{inv}(\mathbf{X}) \triangleq \mathbf{X}^{-1}$ due to using beamforming to avoid multi-user interference unless the SNR is very low, which may be rather non-smooth. To see this, we consider a special case of $f_\mathrm{inv}(\mathbf{X})$ where the input variable is a scalar, i.e., $f_\mathrm{inv}(x)=1/x$. When $x\to 0$, $f_\mathrm{inv}(x)$ is non-smooth, since the derivative of $f_\mathrm{inv}(x)$ is $f^\prime_\mathrm{inv}(x)=-1/x^2\to\infty$, i.e., a small change of $x$ leads to a significant change of $f_\mathrm{inv}(x)$. Unfortunately, in general cases, the input and output of $f_\mathrm{inv}(\mathbf{H})$ are matrices, and when such a multivariate function is non-smooth cannot be observed from its derivative.

\subsubsection{A Metric to Measure the Smoothness of Precoding Policy}
At first glance,  a Lipschitz constant defined as $\sup_{\mathbf{H},\mathbf{H}^\prime}||f^*(\mathbf{H}^\prime)-f^*(\mathbf{H})||/||\mathbf{H}^\prime-\mathbf{H}||$ can be used to measure the smoothness of a precoding policy, since it can be regarded as the first-order difference of $f^\star(\mathbf{H})$ if without the supremum. However, this metric cannot be applied to wireless communications because the input variables of wireless policies are random, then the supremum in the definition is meaningless. Besides, Lipschitz constant quantifies the smoothness of $f^\star(\cdot)$ by the sensitivity of \emph{absolute} change of its output to \emph{absolute} change of its input, and is used in approximation theory to analyze a specific kind of approximation error $\int_\mathbf{H}||f^*(\mathbf{H})-\hat f(\mathbf{H})||\mathrm{d}\mathbf{H}$, i.e., the \emph{absolute} distance between the target and learned functions. However, the performance of a learned wireless policy by a DNN is often measured by its \emph{relative} utility (say sum rate) to a numerical algorithm, since the utility achieved by the optimal solution differs among system settings.

To cope with this issue, we define a \textbf{smoothness coefficient (SC)} to measure the smoothness of the optimal precoding policy $\mathbf{V}=f^\star(\mathbf{H})$ as
\begin{equation}
\mathrm{SC}\triangleq\mathbb{E}\left[\frac{||f^\star(\mathbf{H}^\prime)-f^\star(\mathbf{H})||}{||f^\star(\mathbf{H})||}\bigg/\frac{||\mathbf{H}^\prime-\mathbf{H}||}{||\mathbf{H}||}\right],
\label{eq_def_sc}
\end{equation}
\noindent where $\mathbf{H}^\prime=\mathbf{H}+\Delta\mathbf{H}$, $\Delta\mathbf{H}$ is a small disturbance on $\mathbf{H}$, and the expectation is taken over $\mathbf{H}$ and $\Delta\mathbf{H}$. A policy is non-smooth if its SC is large.

In what follows, we analyze the impact of environment and system parameters on the SC of the precoding policy.

\subsubsection{Impact of SNR on the Smoothness}

The optimal precoding policy $\mathbf{V}^\star=f^\star(\mathbf{H})$ becomes different multivariate functions of $\mathbf{H}$ under different SNRs.

When the SNR is high, i.e., $\sigma^2{\to}0$, the optimal precoder in \eqref{eq_duality} approaches to the zero-forcing beamforming (ZFBF) with power allocation \cite{bjornson2014optimal}, i.e.,
\begin{equation}
\mathbf{V}^\star=\mathbf{H}^\mathsf{H}\left(\mathbf{H}\mathbf{H}^\mathsf{H}\right)^{-1}\boldsymbol{\Lambda}^{-1}\widetilde{\mathbf{P}}^\frac{1}{2},
\label{eq_high_snr}
\end{equation}
\noindent which is a function of $(\mathbf{H}\mathbf{H}^\mathsf{H})^{-1}$. According to the definition of SC and the definition of condition number of a matrix in \cite{trefethen1997numerical}, i.e., $\kappa(\mathbf{H})\triangleq\sup_{\mathbf{H},\mathbf{H}^\prime}\left(\frac{||f_\mathrm{inv}(\mathbf{H}^\prime)-f_\mathrm{inv}(\mathbf{H})||}{||f_\mathrm{inv}(\mathbf{H})||}\bigg/\frac{||\mathbf{H}^\prime-\mathbf{H}||}{||\mathbf{H}||}\right)$, we know that the matrix inverse function is non-smooth if its input matrices are with large condition numbers. Since the values of $\kappa(\mathbf{H}\mathbf{H}^\mathsf{H})$ for random channel realizations may be very large, the policy is non-smooth and thus harder to learn by DNNs, especially when the channels are highly correlated as to be analyzed soon.

When the SNR is moderate, the optimal precoder in \eqref{eq_duality} has a similar structure to the regularized-ZFBF (R-ZFBF) \cite{peel2005vector} with power allocation, i.e.,
\begin{equation}
\mathbf{V}_\mathrm{RZF}=\left(\mathbf{I}_N+\frac{P_\mathrm{t}}{K\sigma^2}\mathbf{H}^\mathsf{H}\mathbf{H}\right)^{-1}\mathbf{H}^\mathsf{H}\mathbf{P}^\frac{1}{2},
\label{RZF}
\end{equation}
\noindent which is a function of $\left(\mathbf{I}_N+\frac{P_\mathrm{t}}{K\sigma^2}\mathbf{H}^\mathsf{H}\mathbf{H}\right)^{-1}$. The policy is easier to learn than the policy in high SNR, because $\mathbf{I}_N+\frac{P_\mathrm{t}}{K\sigma^2}\mathbf{H}^\mathsf{H}\mathbf{H}$ has a smaller condition number than $\mathbf{H}\mathbf{H}^\mathsf{H}$.

When the SNR is low, i.e., $\sigma^2\to\infty$, the optimal precoder in \eqref{eq_duality} approaches to the maximal-ratio-transmission (MRT) with power allocation \cite{bjornson2014optimal}, which is
\begin{equation}
\mathbf{V}^\star=\mathbf{H}^\mathsf{H}\mathbf{P}^\frac{1}{2}.
\label{eq_low_snr}
\end{equation}
\noindent We can see from \eqref{eq_low_snr} that the smoothness of the precoding policy only depends on the power allocation policy, which can be obtained via water filling (WF) as
\begin{equation}
p_k=\left(\mu-\frac{\sigma^2}{||\mathbf{h}_k||^2}\right)^{+},
\label{eq_wf}
\end{equation}
\noindent where $\mu$ is the water level satisfying $\sum_{k=1}^Kp_k=P_\mathrm{t}$. The power allocation policy is non-smooth at low SNR due to the non-differentiable piecewise function $(\cdot)^{+}=\mathrm{max}(\cdot,0)$, which however can be well-learned as to be shown in the next section. This is because the sum rate is insensitive to the power allocation at low SNR.

\subsubsection{Impact of Channel Correlation  on the Smoothness}  To see how the correlation of channel matrix affects the smoothness of the precoding policy, suppose that the channels of two users are highly correlated, say $\mathbf{h}_1\approx\mathbf{h}_2$. Then, it can be easily shown that $\mathbf{h}_1^\mathsf{H}\mathbf{H}^\mathsf{H}$ and $\mathbf{h}_2^\mathsf{H}\mathbf{H}^\mathsf{H}$, which are the first and second row vectors of $\mathbf{H}\mathbf{H}^\mathsf{H}$, are also highly correlated. Consequently, $\kappa(\mathbf{H}\mathbf{H}^\mathsf{H})$ is very large \cite{trefethen1997numerical}. On the other hand, if the channel vectors $\mathbf{h}_k, k=1,\cdots,K$ are mutually orthogonal and with identical norm (i.e., $\mathbf{H}\mathbf{H}^\mathsf{H}=\mathbf{I}$), both $\lambda_k$ and $p_k$ in \eqref{eq_duality} become $P_\mathrm{t}/{K}$ \cite{bjornson2014optimal}, then the precoding policy $f^\star(\mathbf{H})=\sqrt{(P_\mathrm{t}/{K})} \mathbf{H}^\mathsf{H}$ becomes MRT with equal power allocation, which is smooth. This indicates that learning the precoding policy in spatially correlated channels is more challenging than in spatially uncorrelated channels.

In multi-user multi-antenna systems, user scheduling is often used to improve the sum rate or other utilities, say by selecting semi-orthogonal users from candidate users \cite{yoo2006optimality}. Since such a scheduling algorithm can improve the orthogonality of the channel vectors of the scheduled users, it can reduce the approximation error when learning the precoding policy for the scheduled users with a DNN.

\subsubsection{Impact of $N$ and $K$  on the Smoothness}
As proved in \cite{marzetta2010noncooperative}, $\mathbf{h}_k, k=1,\cdots,K$ are asymptotically orthogonal when $N\to\infty$ with fixed $K$. It indicates that when $N \gg K$, the precoding policy is smoother and hence easier to learn by DNNs, even if $N$ and $K$ are very large.

In a nutshell, the answer to Q2 in section \ref{motivation-contribution} is: the
precoding policy is harder to learn by DNNs when the SNR is high, the channels are spatially correlated, or the number of users is comparable to the number of antennas. The hardness comes from the beamforming policy for avoiding multi-user interference, since computing the beamforming matrix needs a matrix inverse operation explicitly or implicitly. This suggests that: to reduce the approximation error (and thus improve the learning performance) of a DNN, an effective approach is to simplify the target function such that it does not consist of matrix inverse function. This can explain why the model-driven DNNs in \cite{kim2020deep,shi2021deep,zhang2021model,yuan2020transfer,xia2020deep,kim2022bipartite} can perform well in optimizing baseband precoding from several problems.

%

\vspace{2mm}\subsection{Which DNN is More Efficient for Learning Precoding?}\label{subsec_rib}

Introducing appropriate inductive biases into the architecture of a DNN can achieve a better trade-off between the approximation and estimation errors. This makes the DNN more efficient, i.e., requiring less training samples and trainable parameters to achieve an expected performance. As shown in Fig. \ref{fig_est_appr}(c), by imposing constraints on the function family able to be represented, one can reduce the hypothesis space of the DNN and thus decrease the estimation error for a given training set. If the imposed inductive bias matches with the property of the policy, then the best-fitted policy $\hat{f}^\star(\cdot)$ still lies in the reduced hypothesis space, which leads to less estimation error without changing the approximation error. Mismatched inductive biases, however, incur larger approximation error, which has to be reduced by enlarging the scale of the DNN.

In the sequel, we prove that the EdgeGNN exhibits matched inductive bias to the precoding policy, but CNN does not. Hence, the EdgeGNN has lower estimation error bound and is more efficient than the CNN and FNN.

\subsubsection{Inductive Biases}\label{hypos}

Denote the precoding policy learned by the EdgeGNN as $\widehat{\mathbf{V}}=\hat{f}_\mathrm{G}(\mathbf{H})$. As shown in the following proposition, it satisfies the same PE property as $\mathbf{V}^\star=f^\star(\mathbf{H})$.

\begin{proposition}\label{prop_pe}
When learning over the precoding graph with the EdgeGNN, $\boldsymbol{\Pi}_1^\mathsf{T}\widehat{\mathbf{V}}\boldsymbol{\Pi}_2=\hat{f}_\mathrm{G}(\boldsymbol{\Pi}_1^\mathsf{T}\mathbf{H}\boldsymbol{\Pi}_2)
$.
\end{proposition}
\begin{proof}\vspace{-2mm}
See Appendix A.
\end{proof}\vspace{-2mm}

To analyze the inductive bias of the CCNN,  we express the input-output relation of a circular convolutional layer in \eqref{eq_forward_circ_cnn} with shifted input as
\begin{IEEEeqnarray*}{l}
\label{eq_shift_equi}
\phi\left(\mathbf{W}_\mathrm{C}\mathrm{vec}(\mathbf{S}_1^\mathsf{T}\mathbf{D}^{(l-1)}\mathbf{S}_2)\right)\\
\quad\overset{\mathrm{(a)}}{=}\phi\left(\mathbf{W}_\mathrm{C}(\mathbf{S}_1\otimes\mathbf{S}_2)\mathrm{vec}(\mathbf{D}^{(l-1)})\right)\\
\quad\overset{\mathrm{(b)}}{=}\phi\left((\mathbf{S}_1\otimes\mathbf{S}_2)\mathbf{W}_\mathrm{C}\mathrm{vec}(\mathbf{D}^{(l-1)})\right)\\
\quad=(\mathbf{S}_1\otimes\mathbf{S}_2)\mathrm{vec}(\mathbf{D}^{(l)})\overset{\mathrm{(c)}}{=}\mathrm{vec}(\mathbf{S}_1^\mathsf{T}\mathbf{D}^{(l)}\mathbf{S}_2),\IEEEeqnarraynumspace\IEEEyesnumber
\end{IEEEeqnarray*}
\noindent where both  $\mathbf{S}_1$ and $\mathbf{S}_2$ are circulant matrices representing the shift operations,
and the multiplication of matrix $\mathbf{S}_1^\mathsf{T}$ (or $\mathbf{S}_2$) with tensor $\mathbf{D}^{(l-1)}\in\mathbb{R}^{K{\times}N{\times}C_{l-1}}$ is obtained by multiplying $\mathbf{S}_1^\mathsf{T}$ (or $\mathbf{S}_2$) with $C_{l-1}$ matrices of size $K{\times}N$ within the tensor.
$(a)$ and $(c)$ hold because $\mathrm{vec}(\mathbf{ABC})=(\mathbf{C}^\mathsf{T}\otimes\mathbf{A})\mathrm{vec}(\mathbf{B})$, and $(b)$ holds because both $\mathbf{S}_1\otimes\mathbf{S}_2$ and $\mathbf{W}_\mathrm{C}$ are circulant matrices whose multiplication is commutative.

By stacking $L$ layers, it is not hard to prove that the policy learned by the CCNN (denoted as $\widehat{\mathbf{V}}=\hat{f}_\mathrm{C}(\mathbf{H})$) is shift equivariant to $\mathbf{H}$, i.e., $\mathbf{S}_1^\mathsf{T}\widehat{\mathbf{V}}\mathbf{S}_2=\hat{f}_\mathrm{C}(\mathbf{S}_1^\mathsf{T}\mathbf{H}\mathbf{S}_2)$.

Since $\mathbf{S}_1$ and $\mathbf{S}_2$ are special permutation matrices, the hypothesis space of the CCNN contains the hypothesis space of the EdgeGNN when $\mathbf{W}_\mathrm{C}$ and $\mathbf{W}_\mathrm{G}$ are with the same size. This can also be seen from previous analysis in section \ref{subsec_relation} that $\mathbf{W}_\mathrm{C}$ degenerates into $\mathbf{W}_\mathrm{G}$ after sharing more parameters.
For the CCNN with strides larger than one,
the shift-equivariance property holds only for specific
shift matrices \cite{zhang2019making},\footnote{Take $\mathbf{H}\in\mathbb{C}^{4\times4}$ for example, $\mathbf{S}_1,\mathbf{S}_2\in\mathcal{S}$ can be any shift matrices with dimension four for single-stride CCNN. When the stride is two, $\mathbf{S}_1$ and $\mathbf{S}_2$ can only take values from $\mathcal{S}^\prime\subsetneq\mathcal{S}$ representing $2m$-unit shifts, $m\in\mathbb{Z}$. A permutation matrix $\boldsymbol{\Pi}$ is provided as follows, from which we can see that shift matrices are special permutation matrices.
\arraycolsep=1.5pt\def\arraystretch{0.3}
\begin{IEEEeqnarray*}{c}
\mathcal{S}=\left\{\left[\begin{array}{cccc}
1 & 0 & 0 & 0\\
0 & 1 & 0 & 0\\
0 & 0 & 1 & 0\\
0 & 0 & 0 & 1
\end{array}\right],\left[\begin{array}{cccc}
0 & 1 & 0 & 0\\
0 & 0 & 1 & 0\\
0 & 0 & 0 & 1\\
1 & 0 & 0 & 0
\end{array}\right],\left[\begin{array}{cccc}
0 & 0 & 1 & 0\\
0 & 0 & 0 & 1\\
1 & 0 & 0 & 0\\
0 & 1 & 0 & 0
\end{array}\right],\left[\begin{array}{cccc}
0 & 0 & 0 & 1\\
1 & 0 & 0 & 0\\
0 & 1 & 0 & 0\\
0 & 0 & 1 & 0
\end{array}\right]\right\},\\
\mathcal{S}^\prime=\left\{\left[\begin{array}{cccc}
1 & 0 & 0 & 0\\
0 & 1 & 0 & 0\\
0 & 0 & 1 & 0\\
0 & 0 & 0 & 1
\end{array}\right],\left[\begin{array}{cccc}
0 & 0 & 1 & 0\\
0 & 0 & 0 & 1\\
1 & 0 & 0 & 0\\
0 & 1 & 0 & 0
\end{array}\right]\right\},\boldsymbol{\Pi}=\left[\begin{array}{cccc}
0 & 0 & 1 & 0\\
1 & 0 & 0 & 0\\
0 & 1 & 0 & 0\\
0 & 0 & 0 & 1
\end{array}\right].\IEEEnonumber\label{eq_eg_per_mat}
\end{IEEEeqnarray*}
\setlength{\arraycolsep}{5pt} } which leads to a larger hypothesis space. As a consequence, the EdgeGNN can strike a better balance between the approximation error and estimation error and hence require less trainable parameters and training samples than CCNN to achieve the same learning performance for precoding.

The LCNN is not shift equivariant and its hypothesis space differs from the CCNN, as illustrated in Fig.~\ref{fig_est_appr}(c). Only when $N$ and $K$ are large such that the border effect is negligible, the LCNN is approximately shift equivariant.
Due to the mismatched inductive bias, the LCNN needs larger scale  (i.e., $L$ and $C_l$) than the CCNN to avoid the approximation error.

In existing works using CNNs to learn precoding policy, fully connected layers are cascaded after LCNN. Then, the hypothesis space of the F-LCNN contains 2D-PE functions since FNN is a universal approximator \cite{hornik1989multilayer} and the approximation error can be reduced, but the number of trainable parameters and hence the estimation error increases.

As a universal approximator, the hypothesis space of FNN contains all those of the DNNs.


\subsubsection{Estimation Errors}
We derive the estimation error bounds of the EdgeGNN, CCNN and FNN by deriving their Rademacher complexities, while the bound of LCNN is the same as CCNN. Rademacher complexity reflects the size of the hypothesis space of a DNN, which depends on the Lipschitz continuity of the activation functions and the norm  of weight matrix in each layer \cite{lin2019generalization}.

Denote $\widetilde{\mathbf{H}}=\left[\mathrm{vec}(\mathbf{H}_1),\cdots,
\mathrm{vec}(\mathbf{H}_M)\right]\in\mathbb{C}^{KN{\times}M}$ as a training set, where $M$ is the number of training samples. According to the structure of the weight matrices in \eqref{eq_structured_parameter} and \eqref{eq_circ_structured_parameter}, it is not hard to obtain the following F-norm bounded condition for the EdgeGNN, CCNN, and FNN.

\begin{lemma}\vspace{-2mm}[\textbf{F-norm bounded condition}]\label{lamma_norm}
If in the EdgeGNN $||\mathbf{O}_\mathrm{G}^{(l)\prime}||,||\mathbf{P}_\mathrm{G}^{(l)}||,||\mathbf{Q}_\mathrm{G}^{(l)}||{\leq}{a_l}/{\sqrt{KN(K+N-1)}}$,
in the CCNN $||\mathbf{W}_\mathrm{ker}^{<i>(l)}||{\leq}a_l/\sqrt{KNC_l}$, and in FNN $||\mathbf{W}_\mathrm{F}^{(l)}||{\leq}a_l$, then $||\mathbf{W}_{(\cdot)}^{(l)}||{\leq}a_l$, where the subscript $(\cdot)$ can be ``G'', ``C'', and ``F'' to represent the EdgeGNN, CCNN, and FNN, respectively.
\end{lemma}\vspace{-1mm}

In the following proposition, we provide the upper bounds of the estimations errors of each DNN.

\begin{proposition}\vspace{-1mm}\label{prop_est_err}
If (i) the activation functions $\phi^{(l)}(\cdot)$ are with Lipschitz constants $L_\phi^{(l)}<\infty$ and satisfy $\phi^{(l)}(0)=0, l=1,\cdots,L$, (ii) the loss function is with Lipschitz constant $L_\mathrm{loss}<\infty$, (iii) the weight matrices $||\mathbf{W}_{(\cdot)}^{(l)}||_\sigma{\leq}s_l$, and (iv) the F-norm bounded condition in Lemma \ref{lamma_norm} holds such that $||\mathbf{W}_{(\cdot)}^{(l)}||{\leq}a_l$, then the estimation error is upper bounded with probability $1-\delta$ as,
\begin{IEEEeqnarray}{l}\label{eq_est_err}
\mathcal{R}(\hat{f}^\star)-\mathcal{R}(\hat{f}){\leq}\IEEEnonumber\\
64M^{-\frac{5}{8}}\left(\frac{4L^2||\widetilde{\mathbf{H}}||\left(\prod_{l=1}^LL_\phi^{(l)}s_l\right)R_{(\cdot)}}{L_\mathrm{loss}}\right)^{\frac{1}{4}}+6\sqrt{\frac{\mathrm{ln}(1/\delta)}{2M}},\IEEEeqnarraynumspace
\end{IEEEeqnarray}
\noindent where $R_\mathrm{G}=\sum_{l=1}^L{9\sqrt{3}C_{l-1}^2C_l^2a_l}/{s_l}$ for the EdgeGNN, $R_\mathrm{C}=\sum_{l=1}^L{H_\mathrm{conv}^2W_\mathrm{conv}^2C_{l}^2C_{l-1}^2a_l}/{s_l}$ for the CCNN, and $R_\mathrm{F}=\sum_{l=1}^L{K^4N^4C_{l-1}^2C_l^2a_l}/{s_l}$ for FNN.
\end{proposition}
\begin{proof}\vspace{-2mm}
See Appendix \ref{app_b}.
\end{proof}\vspace{-2mm}

The conditions (i) and (ii) are satisfied for the commonly used activation functions (e.g., Sigmoid, ReLU, and variations of ReLU) and loss functions for learning wireless policies (e.g., negative sum-rate), respectively.

From Proposition \ref{prop_est_err}, we can obtain that the estimation error bound of FNN is $\mathcal{O}\left((R_\mathrm{F}/R_\mathrm{G})^{1/4}\right)=\mathcal{O}(KN)$ times larger than the EdgeGNN and $\mathcal{O}(KN/\sqrt{H_\mathrm{conv}W_\mathrm{conv}})$ times larger than CNNs, when $K$ or $N$ are large such that the first term in the right hand side of \eqref{eq_est_err} dominates. On the other hand, to achieve the same estimation error bound, FNN requires $\mathcal{O}(K^2N^2)$ and $\mathcal{O}(K^2N^2/H_\mathrm{conv}W_\mathrm{conv})$ times more training samples than the EdgeGNN and CNNs, respectively. The results can be explained by the fact that the number of trainable parameters in $\mathbf{W}_\mathrm{F}$ is $\mathcal{O}(K^2N^2)$ times larger than $\mathbf{W}_\mathrm{G}$ and $\mathcal{O}(K^2N^2/H_\mathrm{conv}W_\mathrm{conv})$ times larger than $\mathbf{W}_\mathrm{C}$, where the numbers of trainable parameters of the EdgeGNN and CCNN can be obtained from \eqref{eq_structured_parameter} and \eqref{eq_circ_structured_parameter}.

In summary, the answer to Q1 in section \ref{motivation-contribution} is: a DNN with matched inductive bias to a policy needs fewer training samples and trainable parameters to achieve the same estimation error bound. Previous analyses show that parameter sharing imposes inductive bias on DNNs, and more parameters will be shared if more PE properties of a policy can be satisfied. This suggests that to reduce the estimation error (and thus improve the learning performance) or reduce the sample complexity and trainable parameters, the permutation property of wireless policies should be exploited when designing a DNN. However, no existing works for wireless communications have ever theoretically analyzed the connection between the PE property and the learning performance or the sample complexity.
Most existing works even did not notice the connection, and hence only a part of permutation properties of the wireless policies were considered \cite{shen2022graph, jiang2021learning}.

\vspace{2mm}\subsection{Extension of the Analysis to Other Precoding Problems}
While we take sum rate maximization precoding in MISO system as an example, our analyses are also applicable to the problems with other utility functions and constraints, since all baseband precoding policies in MISO system exhibit the 2D-PE property. Moreover, the precoding policies in other system settings also have PE properties and involve matrix inverse functions in order to avoid multi-user interference \cite{ICC2009}.

For example, in multi-user MIMO systems where each user is equipped with $N_\mathrm{r}>1$ receive antennas, the channel matrix of $K$ users is $\mathbf{H}\in\mathbb{C}^{KN_\mathrm{r}{\times}N}$ and the precoding matrix for all the users is $\mathbf{V}_\mathrm{M}\in\mathbb{C}^{KN_\mathrm{r}{\times}N}$. Then, the precoding policy $\mathbf{V}_\mathrm{M}^\star=f_\mathrm{M}^\star(\mathbf{H})$ satisfies $(\boldsymbol{\Pi}_1\otimes\boldsymbol{\Pi}_2)^\mathsf{T}\mathbf{V}_\mathrm{M}^\star\boldsymbol{\Pi}_3=f_\mathrm{M}^\star((\boldsymbol{\Pi}_1\otimes\boldsymbol{\Pi}_2)^\mathsf{T}\mathbf{H}\boldsymbol{\Pi}_3)$, where $\boldsymbol{\Pi}_1\in\mathbb{R}^{K{\times}K},\boldsymbol{\Pi}_2\in\mathbb{R}^{N_\mathrm{r}{\times}N_\mathrm{r}}$, and $\boldsymbol{\Pi}_3\in\mathbb{R}^{N{\times}N}$ are arbitrary permutation matrices.

Another example is the hybrid precoding in millimeter-wave MISO system, where a BS with $N$ antennas and $N_\mathrm{RF}$  radio frequency chains serves $K$ single-antenna users. When the analog and baseband precoders $\mathbf{V}_\mathrm{RF}\in\mathbb{C}^{N_\mathrm{RF}{\times}N}$ and $\mathbf{V}_\mathrm{D}\in\mathbb{C}^{K{\times}N_\mathrm{RF}}$ are jointly optimized, the precoding policy $(\mathbf{V}^\star_\mathrm{D},\mathbf{V}^\star_\mathrm{RF})=f_\mathrm{H}^\star(\mathbf{H})$ satisfies $(\boldsymbol{\Pi}_1^\mathsf{T}\mathbf{V}^\star_\mathrm{D}\boldsymbol{\Pi}_2,\boldsymbol{\Pi}_2^\mathsf{T}\mathbf{V}^\star_\mathrm{RF}\boldsymbol{\Pi}_3)=f_\mathrm{H}^\star(\boldsymbol{\Pi}_1^\mathsf{T}\mathbf{H}\boldsymbol{\Pi}_3)$. For hybrid precoding in millimeter-wave MIMO system, the policy is with more complicated PE property.

By judiciously designing a graph and the parameter sharing in each layer for a policy, a GNN can exhibit policy-matched inductive bias, whose estimation error is less than CNN and FNN with the same number of training samples. Since the precoding matrix $\mathbf{V}_\mathrm{M}$ in  multi-user MIMO systems and the baseband precoding matrix $\mathbf{V}_\mathrm{D}$ in millimeter-wave multi-user multi-antenna systems are obtained via matrix inverse operation explicitly or implicitly unless the SNR is very low, the channel spatial correlation as well as the numbers of users and antennas also affect the smoothness of the corresponding policies and hence affect the approximation errors when learning the precoding policies.

\vspace{2mm}\section{Simulation Results}
In this section, we validate our analyses by comparing the system performance of a precoding policy learned by the EdgeGNN, CCNN, LCNN, F-LCNN \cite{kim2020deep} and FNN \cite{kim2020deep}, as well as their training complexity.

To this end, we take the sum rate maximization problem in \eqref{eq_sr_constraint} as an example. To satisfy the power constraint, the output layer of each DNN can be normalized as $\widehat{\mathbf{V}}=\sqrt{P_\mathrm{t}/\mathrm{Tr}\left((\mathbf{D}^{(L)})^\mathsf{H}\mathbf{D}^{(L)}\right)}\mathbf{D}^{(L)}\in\mathbb{R}^{K\times N \times 2}$.

All the DNNs are trained via unsupervised learning with $N_\mathrm{tr}=200,000$ training samples, where the negative sum rate is used as the loss function, and the expression of the sum rate is shown in \eqref{eq_sr_constraint}. The learning performance is evaluated with 2000 test samples. The activation function of hidden layers is leaky ReLU. The initial learning rate is $4\times10^{-3}$ and the batch size is 256. The CNNs are with kernel size $H_\mathrm{conv}=W_\mathrm{conv}=3$. All the DNNs are well-trained\footnote{If a DNN performs good on training set but poorly on validation set, then we can reduce the scale of the DNN or increase the training samples. If a DNN performs poorly on both training and validation sets, then enlarging the DNN scale may be helpful. Otherwise, we need to re-tune other hyper-parameters.} using Adam optimizer for the setting of SNR = 10 dB, $N=8$ and $K=4$, and the other fine-tuned hyper-parameters are provided in Table \ref{tab_hp}, where ``3+1'' or ``3+3'' means three convolutional layers plus one or three fully connected layers. When calculating the SC with \eqref{eq_def_sc}, the elements of $\Delta\mathbf{H}$ are generated according to  $\mathcal{CN}(0,\sigma^2_\Delta)$ where $\sigma_\Delta^2=1\times10^{-4}$, and the expectation in \eqref{eq_def_sc} is replaced by empirical average over 5000 realizations of $\mathbf{H}$ and $\Delta\mathbf{H}$.

\begin{table}[htbp]
\topcaption{Hyper-parameters}
\centering
\begin{tabular}{|c|c|c|c|c|c|c|}
\hline
\multicolumn{2}{|c|}{\textbf{Hyper}} & \multicolumn{5}{c|}{\textbf{Neural Networks}} \\
\cline{3-7}
\multicolumn{2}{|c|}{\textbf{Parameters}} & \textbf{FNN} & \textbf{LCNN} & \textbf{CCNN} & \textbf{EdgeGNN} & \textbf{F-LCNN} \\
\hline
\multirowcell{2}{$N=4$\\$K=2$} & $C_l$ & 128 & 128 & 128 & 128 & 128\\
\cline{2-7}
 & $L$ & 4 & 4 & 4 & 4 & 3+1\\
\cline{2-7}
\hline
\multirowcell{2}{$N=8$\\$K=4$} & $C_l$ & 128 & 128 & 256 & 256 & 128\\
\cline{2-7}
 & $L$ & 6 & 6 & 8 & 8 & 3+3\\
\cline{2-7}
\hline
\end{tabular}
\label{tab_hp}
\end{table}

These setups and hyper-parameters are used in the sequel unless otherwise specified.

\subsection{Comparison with Non-learning Methods}

We first compare the performance of the precoding policies learned by the EdgeGNN with the performance achieved by the beamformers with closed-form expressions (i.e., ZFBF, R-ZFBF and MRT) and WF power allocation as well as the performance achieved by numerical algorithms (i.e., the SCA and WMMSE algorithms initialized by R-ZFBF). We consider the UMa scenario in 3GPP TR 38.901\cite{3gpp38901} with both path loss and shadowing, and the users are randomly located in a macro cell of 500 meter radius.

\begin{figure}[htbp]
\centerline{\includegraphics[width=0.7\linewidth]{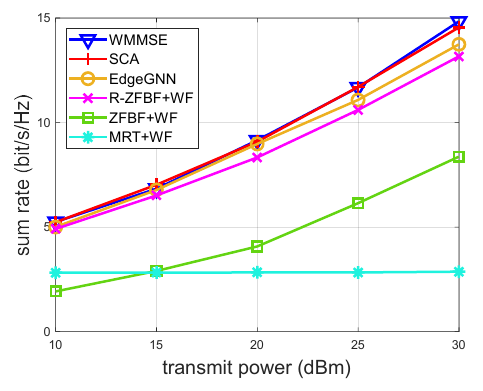}}\vspace{-3mm}
\caption{Sum rate of different methods, $N=8, K=4$.}
\label{fig_sr_snr}
\end{figure} \vspace{-0.2mm}

In Fig. \ref{fig_sr_snr}, we provide the simulation results. We can see that the EdgeGNN has a performance loss from the numerical algorithms only when the transmit power is high and outperforms R-ZFBF+WF, while both ZFBF+WF and MRT+WF exhibit a large performance gap from the EdgeGNN. The CCNN performs close to the EdgeGNN but the FNN, LCNN and F-LCNN are inferior to the EdgeGNN with high transmit power. The results are not provided for a clear figure. However, the inference time of the learning based methods is significantly lower than that of the numerical algorithms. For instance, when the transmit power is 30 dBm, the WMMSE and SCA algorithms take 103 ms and 847 ms in a regular computer, respectively, where the iterations stop when the increment of the sum-rate from a new iteration is less than $1{\times}10^{-3}$ bit/s/Hz. The inference time of the DNNs is provided in Table \ref{tab_inference_time}.

\begin{table}[htb]\vspace{-0.2mm}
\caption{Inference Time, $N=8,K=4,\mathrm{SNR}=30\ \mathrm{dBm}$}\vspace{-2mm}
\centering
\begin{tabular}{|c|c|c|c|c|}
\hline
\textbf{ FNN } & \textbf{ LCNN } & \textbf{ CCNN } & \textbf{ EdgeGNN } & \textbf{ F-LCNN } \\
\hline
0.9 ms & 1.5 ms & 3.0 ms & 3.3 ms & 1.3 ms \\
\hline
\end{tabular}
\label{tab_inference_time}
\end{table}

\subsection{Impact of Inductive Biases and System Settings}
In what follows, we quantify the impact of the inductive biases on both the learning performance and training complexity, and the impact of the key factors analyzed in section \ref{subsec_appr} on the learning performance. We use the following channel model \cite{loyka2001channel} to generate training and test samples,
\begin{equation}
\mathbf{h}_k=\sqrt{\frac{1}{1+\gamma}}\boldsymbol{\Theta}^{\frac{1}{2}}\mathbf{z}_k+\sqrt{\frac{\gamma}{1+\gamma}}\tilde{\mathbf{z}}_k,
\label{eq_rician_model}
\end{equation}
where $\mathbf{z}_k\sim\mathcal{CN}(\mathbf{0}_N,\mathbf{I}_N)$, $\tilde{\mathbf{z}}_k\in\mathbb{C}^{N{\times}1}$ is a deterministic vector, $\gamma \geq 0$ is the Rician factor, $\boldsymbol{\Theta}\in\mathbb{C}^{N{\times}N}$ is the correlation matrix with element in the $i$th row and $j$th column as $\theta_{i,j}=\rho^{j-i}$ for $i \leq j$ and $\theta_{i,j}=\theta_{j,i}$ for $i>j$, and $\rho\in[0,1]$ reflects the correlation of channels. We consider uniform linear arrays, and the $n$th element of $\tilde{\mathbf{z}}_k$ is $\mathrm{exp}(-\mathrm{i}(n-1)\pi\sin(\alpha_k))$, where $\alpha_k\sim\mathcal{U}[-\pi,\pi)$ is the angle of arrival of the direct path.

Unless otherwise specified, we set $\rho=\gamma=0$, i.e., a spatially uncorrelated Rayleigh fading channel model.

\subsubsection{Impact of Inductive Bias}
In Table \ref{tab1}, we show the tested system performance, which is measured by the sum rate achieved by each DNN normalized by the sum rate achieved by the WMMSE algorithm. When $N=4$ and $K=2$, the approximation error of all DNNs are small, and the $200,000$ training samples are sufficient for all the DNNs to achieve more than $99\%$ normalized sum rate in this setting. When $N=8$ and $K=4$, the EdgeGNN achieves the highest sum rate, and the LCNN performs the worst due to the mismatched inductive bias, as expected. Our simulation demonstrates that the performance of the LCNN can only be increased to 94.4$\%$ from 92.8$\%$ even using much more training samples to reduce the estimation error,  which indicates that its performance loss from the EdgeGNN comes from the undiminishable approximation error (not shown due to the lack of space). After cascaded with fully connected layers, the F-LCNN performs close to the FNN. When $N=8$ and $K=4$, the training and test performance of the FNN are respectively 99.4$\%$ and 94.7$\%$, while those of the  F-LCNN are respectively 99.2$\%$ and 94.6$\%$, and those of the EdgeGNN are nearly identical. The gaps between the training and test performance of FNN and F-LCNN come from the large hypothesis space incurred by the fully connected layers.

\vspace{-0.1mm}\begin{table}[htb]
\topcaption{System performance, $\mathrm{SNR}=10\ \mathrm{dB}$, $N_\mathrm{tr}=200,000$}
\centering
\begin{tabular}{|c|c|c|}
\hline
\textbf{Neural} & \multicolumn{2}{c|}{\textbf{Normalized Sum Rate}} \\
\cline{2-3}
\textbf{Networks} & $N=4$, $K=2$ & $N=8$, $K=4$ \\
\hline
\textbf{WMMSE (bits/s/Hz)} & 7.65 & 14.81 \\
\hline
\textbf{FNN} & 99.3$\%$ & 94.7$\%$ \\
\hline
\textbf{LCNN} & 99.1$\%$ & 92.8$\%$ \\
\hline
\textbf{F-LCNN} & 99.4$\%$ & 94.6$\%$ \\
\hline
\textbf{CCNN} & 99.8$\%$ & 98.8$\%$ \\
\hline
\textbf{EdgeGNN} & \textbf{99.9}$\%$ & \textbf{99.0}$\%$ \\
\hline
\end{tabular}
\label{tab1}
\end{table}

In Table \ref{tab4}, we show the training complexity in terms of sample complexity and the scale of the DNNs, which are the number of training samples and the number of free parameters (in million, using `M' for short) required by each DNN to achieve the same normalized sum rate, respectively. It is shown that the CCNN is with higher training complexity than the EdgeGNN due to the larger hypothesis space, and the LCNN needs more training samples than the CCNN due to its mismatched hypothesis space. When $N=4$, $K=2$, the LCNN needs fewer training samples than the FNN and F-LCNN. This is because the approximation error of the LCNN in this case is small, and fewer samples are required to reduce the estimation error thanks to the fewer trainable parameters.
When $N=8$ and $K=4$, the sample complexity of the LCNN is even larger than the FNN in order to compensate for the large approximation error caused by the mismatched hypothesis space.

In summary, the normalized sum rate of the CCNN is lower than the EdgeGNN with a given number of training samples, and the training complexity of the CCNN is higher than the EdgeGNN to achieve the same performance. This indicates that the learning efficiency of the CCNN is lower than the EdgeGNN, which is caused by the non-exactly matched inductive bias to the precoding policy that leads to larger estimation error. This agrees with the analyses in previous section.

\vspace{-2mm}\begin{table}[htbp]
\topcaption{Training Complexity, $\mathrm{SNR}=10\ \mathrm{dB}$}
\centering
\begin{tabular}{|c|c|c|c|c|}
\hline
\textbf{Neural} & \multicolumn{4}{c|}{\textbf{Sample Complexity \& Free Parameters}} \\
\cline{2-5}
\textbf{Networks} & \multicolumn{2}{c|}{$N=4$, $K=2$} & \multicolumn{2}{c|}{$N=8$, $K=4$} \\
\hline
\textbf{System Performance} & \multicolumn{2}{c|}{$95\%$} & \multicolumn{2}{c|}{$90\%$} \\
\hline
\textbf{FNN} & 4500 & 3M & 73,000 & 81M \\
\hline
\textbf{LCNN} & 3500 & 0.43M & 110,000 & 0.7M \\
\hline
\textbf{F-LCNN} & 4300 & 1.3M & 65,000 & 49M \\
\hline
\textbf{CCNN} & 800 & 0.43M & 20,000 & 4M \\
\hline
\textbf{EdgeGNN} & \textbf{400} & \textbf{0.14M} & \textbf{12,000} & \textbf{2M} \\
\hline
\end{tabular}
\label{tab4}
\end{table}

\subsubsection{Impact of SNR, $N$, $K$, and Channel Statistics}

In Fig. \ref{fig_snr}, we provide the tested sum rate versus SNR in the range of $-20$ dB$\sim$30 dB, where SNR is defined as $P_\mathrm{t}\mathbb{E}\left[|h|^2\right]/\sigma^2$ without considering the array gain and $h$ is the channel coefficient.\footnote{We find in two typical scenarios in 3GPP TR 38.901 (i.e., UMa and UMi with carrier frequency $f_\mathrm{c}=6$ GHz) \cite{3gpp38901} that more than 98$\%$ of receive SNRs are within -20 dB and 30 dB, where the array gain is not considered and the inter-cell interference is included in noise.} It is shown that the sum rates achieved by all the DNNs have a large gap from the WMMSE algorithm when SNR is 30 dB, but slightly exceed the algorithm when SNR is -20 dB. This can be explained as follows. In the high SNR regime, the optimal precoder is ZFBF with equal power allocation, where the matrix inverse function is hard to learn. In the low SNR regime, the optimal precoder is MRT+WF. Despite that water-filling power allocation policy cannot be well-learned due to its non-smoothness, the performance is insensitive to the power allocation.

\vspace{-0.2mm}\begin{figure}[htbp]
\centerline{\includegraphics[width=0.7\linewidth]{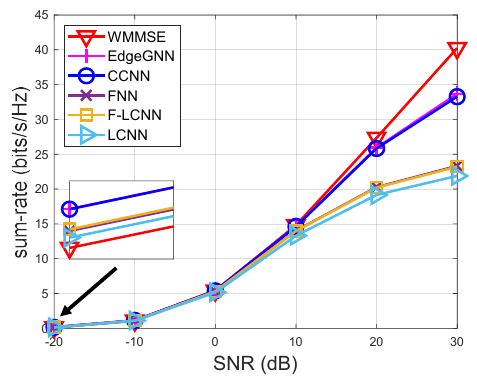}}\vspace{-2mm}
\bottomcaption{Impact of SNR on the learning performance, $N=8$, $K=4$.}
\label{fig_snr}
\end{figure}

To validate the interpretation, we first provide the normalized sum rates of different precoding policies obtained by the non-learning methods in Fig.~\ref{fig_sc_snr}(a). As expected, the WMMSE algorithm achieves similar performance to MRT+WF at low SNR and to ZFBF+WF at high SNR. In Fig.~\ref{fig_sc_snr}(b), we provide the SC of the precoding policy. Since precoding consists of beamforming and power allocation, we can analyze their impacts separately. Specifically, we provide the  SCs of the beamforming policy  and power allocation policy, respectively in Fig.~\ref{fig_sc_snr}(c) and Fig.~\ref{fig_sc_snr}(d), where the two policies are jointly optimized with the WMMSE algorithm.  We can see that the SC of the beamforming policy obtained from the WMMSE algorithm is small and close to that of MRT at low SNR, which increases and tends to that of ZFBF policy when SNR is larger. On the contrary, the SCs of different power allocation policies are large at low SNR and tend to zeros when SNR increases, since they all tend to equal power allocation when SNR is high. As a result, the SC of the precoding policy obtained from the WMMSE algorithm is large and hence the precoding policy is harder to learn for both cases of high and low SNRs, as illustrated in Fig.~\ref{fig_sc_snr}(b). However, the sum rate achieved by the policies learned by the DNNs slightly exceeds that of WMMSE algorithm at low SNR as shown in Fig. \ref{fig_snr}. This is because the considered performance (i.e., the sum rate) is insensitive to the power allocation at low SNR and thus the learning performance is still satisfactory even if the DNNs do not allocate powers in the same way as the optimal solution. These analyses can also explain why the learning efficiency is improved after harnessing the structure of the optimal precoding matrix, where only a power allocation policy requires to be learned.

\begin{figure}[htb]
\centerline{\includegraphics[width=1\linewidth]{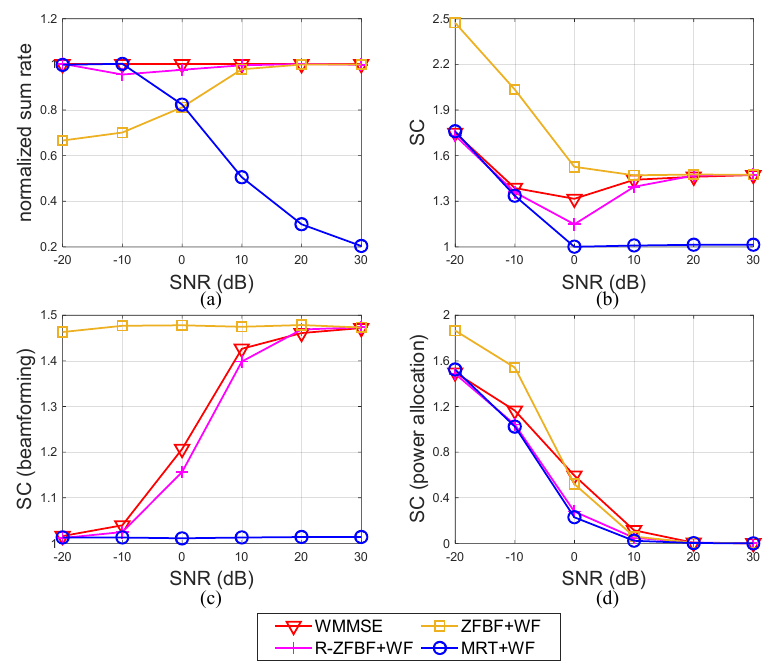}}\vspace{-1mm}
\bottomcaption{(a) Normalized sum rate of the precoding policy versus SNR. (b) Smoothness coefficient of the precoding policy versus SNR. (c) Smoothness coefficient of the beamforming policy versus SNR. (d) Smoothness coefficient of the power allocation policy versus SNR. $N=8$, $K=4$.}
\label{fig_sc_snr}
\end{figure}

In Table \ref{tab2}, we show the impact of the number of users $K$ with a given number of antennas $N$. It is shown that the normalized sum rates of all the DNNs decrease with $K$. This is because the multivariate function $\mathbf{V}^\star=f^\star(\mathbf{H})$ is less smooth and hence harder to learn with more users. To show this, we provide the SC of the precoding policy obtained by the WMMSE algorithm in Table \ref{tab_sc}, which is shown increasing with $K$ and decreasing with $N$.\footnote{It is worth mentioning that if we do not normalize the disturbance in the definition \eqref{eq_def_sc} of SC, then the SC will decrease with $K$ when $K$ is small, which can not explain the trend of the learning performance. This validates the effectiveness of the defined SC.}

\begin{table}[htb]
\topcaption{Impact of the number of users, $N=16$, $\mathrm{SNR}=10\ \mathrm{dB}$}
\centering
\begin{tabular}{|c|c|c|c|c|}
\hline
\textbf{Neural} & \multicolumn{4}{c|}{\textbf{Normalized Sum Rate}} \\
\cline{2-5}
\textbf{Networks} & $K=2$ & $K=4$ & $K=6$ & $K=8$ \\
\hline
\textbf{WMMSE (bits/s/Hz)} & 12.40 & 20.08 & 25.48 & 29.05 \\
\hline
\textbf{FNN} & 98.5$\%$ & 91.4$\%$ & 68.9$\%$ & 55.7$\%$ \\
\hline
\textbf{LCNN} & 98.8$\%$ & 93.8$\%$ & 70.3$\%$ & 56.3$\%$ \\
\hline
\textbf{F-LCNN} & 98.9$\%$ & 94.7$\%$ & 70.8$\%$ & 57.2$\%$ \\
\hline
\textbf{CCNN} & 99.8$\%$ & 99.0$\%$ & 90.0$\%$ & 88.8$\%$ \\
\hline
\textbf{EdgeGNN} & \textbf{99.9}$\%$ & \textbf{99.0}$\%$ & \textbf{94.8}$\%$ & \textbf{91.1}$\%$ \\
\hline
\end{tabular}
\label{tab2}
\end{table}

In Table \ref{tab3}, we show the impact of $N$ with given $K$. For the CCNN with kernel size $H_\mathrm{conv}=W_\mathrm{conv}=3$, its performance degrades when $N=32$ since its receptive field (RF) is $17\times17$, which cannot cover the input of size $4\times32$.\footnote{Receptive field is the size of a region in the input image of a CNN that affects the output of the CNN. The RF of a single-stride CNN is $\mathrm{RF}=(LH_\mathrm{conv}-L+1){\times}(LW_\mathrm{conv}-L+1)$ \cite{araujo2019computing}.
Since each element of the precoding matrix depends on all elements in the channel matrix, a DNN to learn the precoding policy should provide global RF.
For a single-stride CNN, the RF will be smaller than the size of input matrix when $K$ and $N$ are large if the kernel size and $L$ are small, which leads to performance degradation for learning precoding. For the EdgeGNN, the RF is global when $L \geq 2$, which can be easily proved.} The performance of the LCNN degrades when $N=16$ for the same reason. By using a larger kernel size $H_\mathrm{conv}=W_\mathrm{conv}=7$ to achieve global RF, the normalized sum rate of the CCNN increases with $N$ monotonically. The normalized sum rate of the EdgeGNN also grows with $N$ slightly. This is because when $N$ is large, the channels tend to be orthogonal and the policy is smoother as shown in Table \ref{tab_sc}. The FNN has poor performance in both small and large values of $N$ due to different reasons. When $N=4$, the training performance of the FNN is $77.0\%$ since the policy is non-smooth  (hence the approximation error is large) for the case of $N=K$. When  $N=32$, the training performance of the FNN is $92.0\%$ but the test performance is only $58.3\%$, which implies overfitting  (i.e., the estimation error is large).
The trend for the F-LCNN is similar to the FNN for the same reason but with better test performance owing to smaller hypothesis space.

\begin{table}[htb]
\topcaption{Smoothness coefficient versus $K$ and $N$,  WMMSE, $\mathrm{SNR}=10\ \mathrm{dB}$}
\centering
\begin{tabular}{|c|c|c|c|c|}
\hline
\multirowcell{2}{$N=16$} & $K=2$ & $K=4$ & $K=6$ & $K=8$ \\
\cline{2-5}
& \textbf{1.1} & 1.2 & 1.4 & 1.5 \\
\hline
\multirowcell{2}{$K=4$} & $N=4$ & $N=8$ & $N=16$ & $N=32$ \\
\cline{2-5}
& 2.2 & 1.5 & 1.2 & \textbf{1.1} \\
\hline
\end{tabular}
\label{tab_sc}
\end{table}

\begin{table}[htb]
\topcaption{Impact of the number of antennas, $K=4$, $\mathrm{SNR}=10\ \mathrm{dB}$}
\centering
\begin{tabular}{|c|c|c|c|c|}
\hline
\textbf{Neural} & \multicolumn{4}{c|}{\textbf{Normalized Sum Rate}} \\
\cline{2-5}
\textbf{Networks} & $N=4$ & $N=8$ & $N=16$ & $N=32$ \\
\hline
\textbf{WMMSE (bits/s/Hz)} & 9.89 & 14.81 & 20.08 & 24.71 \\
\hline
\textbf{FNN} & 76.9$\%$ & 94.7$\%$ & 91.4$\%$ & 58.3$\%$ \\
\hline
\textbf{LCNN ($3\times3$, $L=6$)} & 75.4$\%$ & 92.8$\%$ & 90.2$\%$ & 89.4$\%$ \\
\hline
\textbf{F-LCNN} & 76.9$\%$ & 94.6$\%$ & 94.7$\%$ & 92.1$\%$ \\
\hline
\textbf{CCNN ($3\times3$, $L=8$)} & \textbf{98.6}$\%$ & 98.8$\%$ & 99.0$\%$ & 95.9$\%$ \\
\hline
\textbf{CCNN  ($7\times7$, $L=8$)} & \textbf{98.6}$\%$ & 98.9$\%$ & 99.0$\%$ & 99.1$\%$ \\
\hline
\textbf{EdgeGNN} & \textbf{98.6}$\%$ & \textbf{99.0}$\%$ & \textbf{99.2}$\%$ & \textbf{99.2}$\%$ \\
\hline
\multicolumn{5}{l}{$3\times3$ and $7\times7$ refer to the kernel sizes of LCNN and CCNN.}
\end{tabular}
\label{tab3}
\end{table}

The orthogonality of channel vectors can be improved by selecting $K$ users from $K_\mathrm{max}$ candidates. With the scheduling algorithm in \cite{yoo2006optimality}, the test  performance of the DNNs versus $K_\mathrm{max}$ is given in Table \ref{tab_kmax}. It is shown that the learning performance increases with $K_\mathrm{max}$, but only slightly.

In Table \ref{table_learning_correlation}, we show the impact of channel correlation. We consider different values of $\gamma$ and $\rho$ in \eqref{eq_rician_model}, which reflect the strength of the direct path and the correlation among the diffused paths, respectively. It is shown that the EdgeGNN and CCNN perform closely. All the DNNs perform the worst in the highly correlated channel without a direct path (i.e., $(\rho,\gamma)=(0.8,0)$). This is because the precoding policy is rather non-smooth in such channels as shown in Fig. \ref{fig_sc_rician}.

\begin{table}[htb]
\topcaption{Impact of user scheduling, $N=8$, $K=4$, $\mathrm{SNR}=10\ \mathrm{dB}$}
\centering
\begin{tabular}{|c|c|c|c|c|c|c|}
\hline
& \multicolumn{6}{c|}{\textbf{Normalized Sum Rate}} \\
\cline{2-7}
\textbf{Neural} & \multicolumn{3}{c|}{\textbf{Rayleigh Channel}} & \multicolumn{3}{c|}{\textbf{Correlated Channel,}} \\
& \multicolumn{3}{c|}{$(\rho,\gamma)=(0,0)$} & \multicolumn{3}{c|}{$(\rho,\gamma)=(0.8,0)$} \\
\cline{2-7}
\textbf{Networks} & \multicolumn{3}{c|}{$K_\mathrm{max}$} & \multicolumn{3}{c|}{$K_\mathrm{max}$} \\
\cline{2-7}
& $4$ & $8$ & $16$ & $4$ & $8$ & $16$\\
\hline
\textbf{FNN} & 94.7$\%$ & 95.0$\%$ & 95.2$\%$ & 78.1$\%$ & 79.1$\%$ & 79.7$\%$\\
\hline
\textbf{LCNN} & 92.8$\%$ & 93.8$\%$ & 94.3$\%$ & 86.7$\%$ & 89.0$\%$ & 89.6$\%$\\
\hline
\textbf{F-LCNN} & 94.6$\%$ & 94.8$\%$ & 95.0$\%$ & 87.4$\%$ & 90.0$\%$ & 90.5$\%$\\
\hline
\textbf{CCNN} & 98.8$\%$ & 99.1$\%$ & 99.1$\%$ & 97.6$\%$ & 98.2$\%$ & 98.5$\%$\\
\hline
\textbf{EdgeGNN} & \textbf{99.0}$\%$ & \textbf{99.2}$\%$ & \textbf{99.3}$\%$ & \textbf{97.8}$\%$ & \textbf{98.8}$\%$ & \textbf{99.1}$\%$\\
\hline
\end{tabular}
\label{tab_kmax}
\end{table}

\begin{table}[htb]
\topcaption{Impact of channels, $N=8$, $K=4$, $\mathrm{SNR}=10\ \mathrm{dB}$}
\centering
\begin{tabular}{|c|c|c|c|c|}
\hline
\textbf{Neural} & \multicolumn{4}{c|}{\textbf{Normalized Sum Rate}} \\
\cline{2-5}
& $(\rho,\gamma)=$ & $(\rho,\gamma)=$ & $(\rho,\gamma)=$ & $(\rho,\gamma)=$\\
\textbf{Networks} & $(0,0)$ & $(0,2)$ & $(0.8,2)$ & $(0.8,0)$\\
\hline
\textbf{WMMSE (bits/s/Hz)} & 14.81 & 14.59 & 14.37 & 11.33 \\
\hline
\textbf{FNN} & 94.7$\%$ & 95.2$\%$ & 95.0$\%$ & 78.1$\%$\\
\hline
\textbf{LCNN} & 92.8$\%$ & 93.1$\%$ & 92.7$\%$ & 86.7$\%$ \\
\hline
\textbf{F-LCNN} & 94.6$\%$ & 95.1$\%$ & 94.6$\%$ & 87.4$\%$\\
\hline
\textbf{CCNN} & 98.8$\%$ & 98.3$\%$ & 98.3$\%$ & 97.6$\%$\\
\hline
\textbf{EdgeGNN} & \textbf{99.0}$\%$ & \textbf{98.8}$\%$ & \textbf{98.6}$\%$ & \textbf{97.8}$\%$\\
\hline
\end{tabular}
\label{table_learning_correlation}
\end{table}

\begin{figure}[htb]
\centerline{\includegraphics[width=0.7\linewidth]{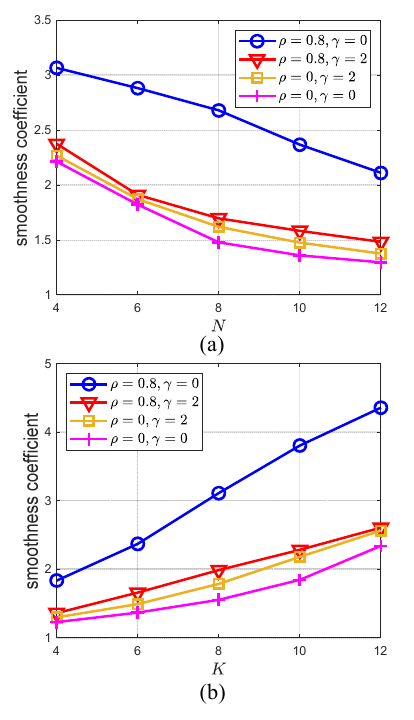}}\vspace{-2mm}
\bottomcaption{Smoothness coefficient, SNR=10 dB. (a) $K=4$. (b) $N=16$.}
\label{fig_sc_rician}
\end{figure}

\subsection{Out-of-distribution Generalization Ability}
All previous analyses and simulations are obtained when the training and testing samples are generated with the same distribution. In the sequel, we compare the generalizability of the well-trained DNNs to the numbers of users and antennas, SNR, and channel statistics, by testing their performance in the environments that are different from training.
\subsubsection{Generalizability to $K$ and $N$}
We first show the the generalizability to the problem scales, where $N_\mathrm{tr}=400,000$.

In Fig. \ref{fig_dim_gen_k}, we show the generalization ability of the EdgeGNN and CCNN to $K$. Both DNNs are trained with the samples generated under the settings of $K\in\{6,7,8\}$, while the test samples are generated with $K\in\{2,\cdots,12\}$. It can be observed that the performance of both DNNs degrade obviously when $K$ of the test samples differs significantly from that of training. The size-generalization performance of CCNN is worse, since its receptive field is fixed for a given kernel size. Furthermore, we can see that the generalizability to $K$ is better when SNR is lower (e.g., 0 dB), which is the case considered in \cite{jiang2021learning} (i.e., $-4$ dB$\sim$5 dB). This is because the precoding policy becomes MRT with WF power allocation at low SNR level. MRT is a single-user beamforming policy and hence can be generalized to multiple users, while power allocation is easier to be generalized to users as evaluated by existing works \cite{eisen2020optimal,guo2021learning}. With higher SNR, the performance of the vertex-GNN designed in \cite{jiang2021learning} also degrades significantly.

\begin{figure}[htb]
\centerline{\includegraphics[width=0.8\linewidth]{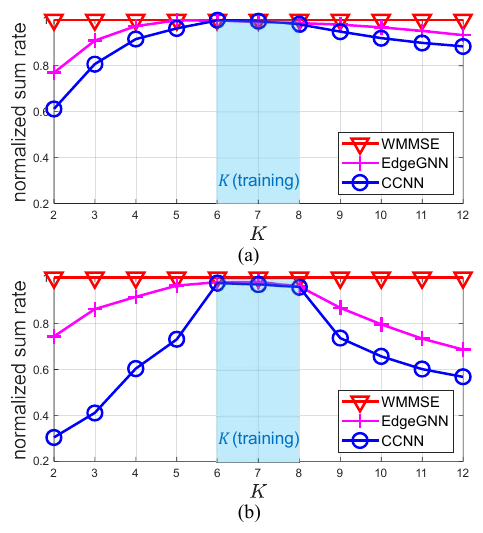}}\vspace{-4mm}
\caption{Sum rate versus $K$, $N=16$. (a) SNR = 0 dB. (b) SNR = 10 dB.}
\label{fig_dim_gen_k}
\end{figure}

In Fig. \ref{fig_dim_gen_n}, we show the generalization ability of the EdgeGNN and CCNN to $N$. Both DNNs are trained with the samples generated under the settings of $N\in\{6,7,8\}$, while the test samples are generated with $N\in\{4,\cdots,14\}$. It can be observed that both the EdgeGNN and the CCNN perform well when generalized to $N$.

\begin{figure}[htb]
\centerline{\includegraphics[width=0.8\linewidth]{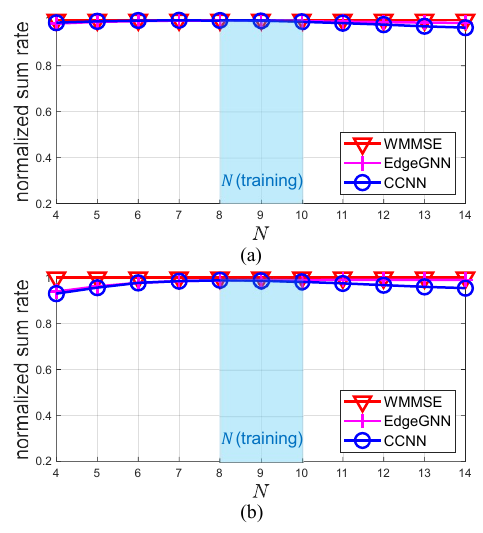}}\vspace{-4mm}
\caption{Sum rate versus $N$, $K=4$. (a) SNR = 0 dB. (b) SNR = 10 dB.}
\label{fig_dim_gen_n}
\end{figure}

\subsubsection{Generalizability to SNR and Channel Statistics}

In Table \ref{table_gen_snr}, we show the generalization ability of the DNNs to different SNRs. For comparison, we also train corresponding DNNs with a \textit{mixed} training set that consists of samples with various SNRs (i.e., SNR$\in\{0,5,10,15,20\}$ dB), where 40,000 training samples are generated for each SNR and a total number of $N_\mathrm{tr}=200,000$ samples are used, and SNR is also inputted to the DNNs.

\begin{table*}[htb]
\footnotesize
\caption{Generalization to SNR, $N=8$, $K=4$, $\rho=\gamma=0$.}
\centering
\begin{tabular}{|c|c|c|c|c|c|c|c|}
\hline
\multicolumn{2}{|c|}{\textbf{SNR of}} & \multicolumn{5}{c|}{\textbf{SNR of Test (dB)}} & \textbf{Average} \\
\cline{3-7}
\multicolumn{2}{|c|}{\textbf{Training (dB)}} & 0 & 5 & 10 & 15 & 20 & \textbf{Performance}\\
\hline
\multirowcell{3}{\textbf{EdgeGNN}} & 0 & \textbf{99.7}$\%$ & 95.4$\%$ & 85.4$\%$ & 73.0$\%$ & 61.0$\%$ & \textbf{82.9}$\%$\\
\cline{2-8}
 & 20 & 83.8$\%$ & 92.9$\%$ & 97.2$\%$ & \textbf{97.0}$\%$ & \textbf{94.7}$\%$ & \textbf{93.1}$\%$ \\
\cline{2-8}
 & Mixed & 98.7$\%$ & \textbf{98.6}$\%$ & \textbf{98.3}$\%$ & 96.6$\%$ & 92.8$\%$ & \textbf{97.0}$\%$ \\
\hline
\multirowcell{3}{\textbf{CCNN}} & 0 & \textbf{99.6}$\%$ & 95.4$\%$ & 85.3$\%$ & 72.6$\%$ & 60.5$\%$ & 82.7$\%$\\
\cline{2-8}
 & 20 & 83.6$\%$ & 92.8$\%$ & 97.0$\%$ & \textbf{97.0}$\%$ & \textbf{94.3}$\%$ & 92.9$\%$\\
\cline{2-8}
 & Mixed & 97.9$\%$ & \textbf{98.4}$\%$ & \textbf{98.2}$\%$ & 96.3$\%$ & 92.5$\%$ & 96.7$\%$ \\
\hline
\multirowcell{3}{\textbf{FNN}} & 0 & \textbf{97.8}$\%$ & 92.5$\%$ & 80.8$\%$ & 67.1$\%$ & 54.9$\%$ & 78.6$\%$\\
\cline{2-8}
 & 20 & 85.7$\%$ & 87.5$\%$ & 85.1$\%$ & 80.3$\%$ & 73.8$\%$ & 82.5$\%$\\
\cline{2-8}
 & Mixed & 97.2$\%$ & \textbf{96.6}$\%$ & \textbf{93.0}$\%$ & \textbf{84.9}$\%$ & \textbf{75.3}$\%$ & 89.4$\%$ \\
\hline
\multirowcell{3}{\textbf{F-LCNN}} & 0 & \textbf{98.0}$\%$ & 92.7$\%$ & 81.1$\%$ & 67.2$\%$ & 55.2$\%$ & 78.8$\%$\\
\cline{2-8}
 & 20 & 85.6$\%$ & 87.5$\%$ & 85.0$\%$ & 80.2$\%$ & 74.4$\%$ & 82.5$\%$\\
\cline{2-8}
 & Mixed & 95.2$\%$ & \textbf{95.6}$\%$ & \textbf{93.0}$\%$ & \textbf{86.5}$\%$ & \textbf{78.0}$\%$ & 89.7$\%$ \\
\hline
\end{tabular}
\label{table_gen_snr}
\end{table*}

It can be observed that the DNNs trained with the samples generated at low SNR (say 0 dB) do not perform well when they are tested at high SNR (say larger than 10 dB), while the DNNs trained with high SNR perform fairly well unless they are tested at very low SNR (i.e., 0 dB). As expected, the DNNs trained with the mixed data set perform well in all the test sets, since the DNNs have ``seen'' all possible samples during training and hence do not need to be generalized. In all cases, the EdgeGNN and CCNN have similar generalization performance, both outperform the FNN and F-LCNN. This is because they have smaller hypothesis spaces.

In Table \ref{table_generalization_edgegnn}, we show the generalization ability of the EdgeGNN to channels with different statistics. The results for other DNNs are similar and hence are not shown. As a comparison, we also train the EdgeGNN with mixed samples generated with $(\rho,\gamma)\in\{(0,0),(0,2),(0.8,2),(0.8,0)\}$, where 50,000 training samples are generated for each pair of $(\rho,\gamma)$ and a total number of $200,000$ samples are used.
As expected, the EdgeGNN trained under a single value of $(\rho,\gamma)$ performs the best when it is tested on the samples with the same statistics, and the average performance of the EdgeGNN is improved by training with the mixed data set. The average performance is high when the EdgeGNN is trained under $(\rho,\gamma)=(0,0)$, which is close to the performance trained with the mixed data set, but is lowest when the EdgeGNN is trained with samples following $(\rho,\gamma)=(0.8,0)$. This is because the channels are distributed widely under the uncorrelated Rayleigh channel where $(\rho,\gamma)=(0,0)$. As a result, the training set generated under $(\rho,\gamma)=(0,0)$ can  ``cover'' the test set of other channels with high probability.

\begin{table}[htb]
\footnotesize
\caption{Generalization to Channels, $N=8$, $K=4$, $\mathrm{SNR}=10\ \mathrm{dB}$.}\vspace{-2mm}
\centering
\begin{tabular}{|c|c|c|c|c|c|}
\hline
$(\rho,\gamma)$ \textbf{of} & \multicolumn{4}{c|}{$(\rho,\gamma)$ \textbf{of Test}} & \textbf{Average} \\
\cline{2-5}
\textbf{Training} & $(0,0)$ & $(0,2)$ & $(0.8,2)$ & $(0.8,0)$ & \textbf{Performance}\\
\hline
$(0,0)$ & \textbf{99.0}$\%$ & 98.6$\%$ & 98.0$\%$ & 94.7$\%$ & \textbf{97.6}$\%$\\
\hline
$(0,2)$ & 98.9$\%$ & \textbf{98.8}$\%$ & 97.9$\%$ & 91.2$\%$ & 96.7$\%$\\
\hline
$(0.8,2)$ & 97.1$\%$ & 97.3$\%$ & \textbf{98.6}$\%$ & 94.6$\%$ & 96.9$\%$\\
\hline
$(0.8,0)$ & 87.0$\%$ & 83.2$\%$ & 89.4$\%$ & \textbf{97.8}$\%$ & 89.4$\%$\\
\hline
Mixed & 98.5$\%$ & 98.4$\%$ & 98.7$\%$ & 97.3$\%$ & \textbf{98.2}$\%$ \\
\hline
\end{tabular}
\label{table_generalization_edgegnn}
\end{table}

\vspace{3mm}\section{Conclusion}
In this paper, we analyzed when and why the DNNs
for learning precoding policy can perform well. Specifically, we analyzed the impact of inductive biases of DNNs on the performance of learning precoding by designing an edge-GNN, comparing it with two CNNs, and deriving the estimation error bounds of the EdgeGNN, CNN and FNN. We identified in which scenarios precoding policy is harder to be approximated by DNNs via analyzing the smoothness of the policy. Our analyses indicate that the learning efficiency of the EdgeGNN is higher owing to the policy-matched inductive bias, and the precoding policy is harder to learn when the SNR is high, the channels are spatially correlated, or the number of users is comparable to the number of antennas. Simulation results validated the analyses and quantified the impact of the inductive bias of a DNN and the smoothness of a policy on the learning performance and training complexity. We believe that
these analyses can provide useful insight into the design of DNNs for learning precoding policy. For example, when designing a model-driven DNN, one can introduce mathematical models such that the DNN does not need to learn a non-smooth function such as matrix inverse. Nonetheless, our analyses cannot explain the out-of-distribution generalization ability of DNNs, since the theories we used for the analyses are only applicable when the test samples are drawn from the same distribution as the training samples.

\vspace{3mm}
\bibliographystyle{IEEEtran}
\bibliography{Reference}

\begin{thebibliography}{10}
\providecommand{\url}[1]{#1}
\csname url@samestyle\endcsname
\providecommand{\newblock}{\relax}
\providecommand{\bibinfo}[2]{#2}
\providecommand{\BIBentrySTDinterwordspacing}{\spaceskip=0pt\relax}
\providecommand{\BIBentryALTinterwordstretchfactor}{4}
\providecommand{\BIBentryALTinterwordspacing}{\spaceskip=\fontdimen2\font plus
\BIBentryALTinterwordstretchfactor\fontdimen3\font minus
  \fontdimen4\font\relax}
\providecommand{\BIBforeignlanguage}[2]{{%
\expandafter\ifx\csname l@#1\endcsname\relax
\typeout{** WARNING: IEEEtran.bst: No hyphenation pattern has been}%
\typeout{** loaded for the language `#1'. Using the pattern for}%
\typeout{** the default language instead.}%
\else
\language=\csname l@#1\endcsname
\fi
#2}}
\providecommand{\BIBdecl}{\relax}
\BIBdecl

\bibitem{zhao2022learning}
B.~Zhao, J.~Guo, and C.~Yang, ``Learning precoding policy: {CNN} or {GNN}?''
  \emph{IEEE WCNC}, 2022.

\bibitem{boccardi2014mimo}
F.~Boccardi, R.~Heath, A.~Lozano, T.~Marzetta \emph{et~al.}, ``Five disruptive
  technology directions for 5{G},'' \emph{IEEE Commun. Mag.}, vol.~52, no.~2,
  pp. 74--80, Feb. 2014.

\bibitem{tran2012fast}
L.~Tran, M.~Hanif, A.~Tolli, and M.~Juntti, ``Fast converging algorithm for
  weighted sum rate maximization in multicell {MISO} downlink,'' \emph{IEEE
  Signal Process. Lett.}, vol.~19, no.~12, pp. 872--875, Oct. 2012.

\bibitem{shi2011iteratively}
Q.~Shi, M.~Razaviyayn, Z.~Luo, and C.~He, ``An iteratively weighted {MMSE}
  approach to distributed sum-utility maximization for a {MIMO} interfering
  broadcast channel,'' \emph{IEEE Trans. Signal Process.}, vol.~59, no.~9, pp.
  4331--4340, Apr. 2011.

\bibitem{zhang2019deep}
C.~Zhang, P.~Patras, and H.~Haddadi, ``Deep learning in mobile and wireless
  networking: A survey,'' \emph{IEEE Commun. Surv. Tutor.}, vol.~21, no.~3, pp.
  2224--2287, Mar. 2019.

\bibitem{hu2020iterative}
Q.~Hu, Y.~Cai, Q.~Shi, K.~Xu \emph{et~al.}, ``Iterative algorithm induced
  deep-unfolding neural networks: Precoding design for multiuser {MIMO}
  systems,'' \emph{IEEE Trans. Wireless Commun.}, vol.~20, no.~2, pp.
  1394--1410, Feb. 2021.

\bibitem{kim2021learning}
J.~Kim, H.~Lee, and S.~Park, ``Learning robust beamforming for {MISO} downlink
  systems,'' \emph{IEEE Commun. Lett.}, vol.~25, no.~6, pp. 1916--1920, June
  2021.

\bibitem{elbir2021family}
A.~Elbir, K.~Mishra, B.~Shankar, and B.~Ottersten, ``A family of deep learning
  architectures for channel estimation and hybrid beamforming in multi-carrier
  mm-wave massive {MIMO},'' \emph{IEEE Trans. Cogn. Commun. Netw.}, vol.~8,
  no.~2, pp. 642--656, Dec. 2021.

\bibitem{attiah2022deep}
K.~Attiah, F.~Sohrabi, and W.~Yu, ``Deep learning for channel sensing and
  hybrid precoding in {TDD} massive {MIMO} {OFDM} systems,'' \emph{IEEE Trans.
  Wireless Commun.}, vol.~21, no.~12, pp. 10\,839--10\,853, July 2022.

\bibitem{zhang2021model}
J.~Zhang, M.~You, G.~Zheng, I.~Krikidis \emph{et~al.}, ``Model-driven learning
  for generic {MIMO} downlink beamforming with uplink channel information,''
  \emph{IEEE Trans. Wireless Commun}, vol.~21, no.~4, pp. 2368--2382, Apr.
  2022.

\bibitem{yuan2020transfer}
Y.~Yuan, G.~Zheng, K.~Wong, B.~Ottersten \emph{et~al.}, ``Transfer learning and
  meta learning-based fast downlink beamforming adaptation,'' \emph{IEEE Trans.
  Wireless Commun.}, vol.~20, no.~3, pp. 1742--1755, Mar. 2021.

\bibitem{kim2020deep}
J.~Kim, H.~Lee, S.~Hong, and S.~Park, ``Deep learning methods for universal
  {MISO} beamforming,'' \emph{IEEE Wireless Commun. Lett.}, vol.~9, no.~11, pp.
  1894--1898, Nov. 2020.

\bibitem{shi2021deep}
J.~Shi, W.~Wang, X.~Yi, X.~Gao \emph{et~al.}, ``Deep learning-based robust
  precoding for massive {MIMO},'' \emph{IEEE Trans. Commun.}, vol.~69, no.~11,
  pp. 7429--7443, Aug. 2021.

\bibitem{chen2020sub}
K.~Chen, J.~Yang, Q.~Li, and X.~Ge, ``Sub-array hybrid precoding for massive
  {MIMO} systems: A {CNN}-based approach,'' \emph{IEEE Commun. Lett.}, vol.~25,
  no.~1, pp. 191--195, Jan. 2021.

\bibitem{bo2021deep}
Z.~Bo, R.~Liu, M.~Li, and Q.~Liu, ``Deep learning based efficient symbol-level
  precoding design for {MU-MISO} systems,'' \emph{IEEE Trans. Vehicular
  Technol.}, vol.~70, no.~8, pp. 8309--8313, June 2021.

\bibitem{lei2021ci}
Z.~Lei, X.~Liao, Z.~Gao, and A.~Li, ``{CI-NN}: A model-driven deep
  learning-based constructive interference precoding scheme,'' \emph{IEEE
  Commun. Lett.}, vol.~25, no.~6, pp. 1896--1900, Feb. 2021.

\bibitem{jiang2021learning}
T.~Jiang, H.~V. Cheng, and W.~Yu, ``Learning to reflect and to beamform for
  intelligent reflecting surface with implicit channel estimation,'' \emph{IEEE
  J. Sel. Areas Commun.}, vol.~39, no.~7, pp. 1931--1945, May 2021.

\bibitem{kang2022mixed}
K.~Kang, Q.~Hu, Y.~Cai, G.~Yu \emph{et~al.}, ``Mixed-timescale deep-unfolding
  for joint channel estimation and hybrid beamforming,'' \emph{IEEE J. Sel.
  Areas Commun.}, vol.~40, no.~9, pp. 2510--2528, July 2022.

\bibitem{he2020model}
Y.~He, H.~He, C.~Wen, and S.~Jin, ``Model-driven deep learning for massive
  multiuser {MIMO} constant envelope precoding,'' \emph{IEEE Wireless Commun.
  Lett.}, vol.~9, no.~11, pp. 1835--1839, June 2020.

\bibitem{xia2020deep}
W.~Xia, G.~Zheng, Y.~Zhu, J.~Zhang \emph{et~al.}, ``A deep learning framework
  for optimization of {MISO} downlink beamforming,'' \emph{IEEE Trans.
  Commun.}, vol.~68, no.~3, pp. 1866--1880, Mar. 2020.

\bibitem{kim2022bipartite}
J.~Kim, H.~Lee, S.~Hong, and S.~Park, ``A bipartite graph neural network
  approach for scalable beamforming optimization,'' \emph{IEEE Trans. Wireless
  Commun.}, vol.~22, no.~1, pp. 333--347, July 2022.

\bibitem{battaglia2018relational}
P.~Battaglia, J.~Hamrick, V.~Bapst, A.~Gonzalez \emph{et~al.}, ``Relational
  inductive biases, deep learning, and graph networks,''
  \emph{arXiv:1806.01261}, 2018.

\bibitem{baxter2000model}
J.~Baxter, ``A model of inductive bias learning,'' \emph{J. Art. Intell. Res.},
  vol.~12, pp. 149--198, Mar. 2000.

\bibitem{lecun1989backprop}
Y.~LeCun, B.~Boser, J.~Denker, D.~Henderson \emph{et~al.}, ``Backpropagation
  applied to handwritten zip code recognition,'' \emph{Neural Comput.}, vol.~1,
  no.~4, pp. 541--551, Dec. 1989.

\bibitem{sun2022improving}
C.~Sun, J.~Wu, and C.~Yang, ``Improving learning efficiency for wireless
  resource allocation with symmetric prior,'' \emph{IEEE Wireless Commun.},
  vol.~29, no.~2, pp. 162--168, Apr. 2022.

\bibitem{liu2023multidimensional}
S.~Liu, J.~Guo, and C.~Yang, ``Multidimensional graph neural networks for
  wireless communications (early access),'' \emph{IEEE Trans. Wireless
  Commun.}, Aug 2023.

\bibitem{eisen2020optimal}
M.~Eisen and A.~Ribeiro, ``Optimal wireless resource allocation with random
  edge graph neural networks,'' \emph{IEEE Trans. Signal Process.}, vol.~68,
  pp. 2977--2991, Apr. 2020.

\bibitem{guo2021learning}
J.~Guo and C.~Yang, ``Learning power allocation for multi-cell-multi-user
  systems with heterogeneous graph neural networks,'' \emph{IEEE Trans.
  Wireless Commun.}, vol.~21, no.~2, pp. 884--897, Feb. 2022.

\bibitem{shalev2014understanding}
S.~Shalev-Shwartz and S.~Ben-David, \emph{Understanding machine learning: From
  theory to algorithms}.\hskip 1em plus 0.5em minus 0.4em\relax Cambridge
  university press, May 2014.

\bibitem{anthony1999neural}
M.~Anthony and P.~Bartlett, \emph{Neural network learning: Theoretical
  foundations}.\hskip 1em plus 0.5em minus 0.4em\relax cambridge university
  press, Nov. 1999, vol.~9.

\bibitem{kearns1994introduction}
M.~Kearns and U.~Vazirani, \emph{An introduction to computational learning
  theory}.\hskip 1em plus 0.5em minus 0.4em\relax MIT press, 1994.

\bibitem{li2018tighter}
X.~Li, J.~Lu, Z.~Wang, J.~Haupt, and T.~Zhao, ``On tighter generalization bound
  for deep neural networks: {CNNs}, resnets, and beyond,''
  \emph{arXiv:1806.05159}, 2018.

\bibitem{garg2020generalization}
V.~Garg, S.~Jegelka, and T.~Jaakkola, ``Generalization and representational
  limits of graph neural networks,'' \emph{ICML}, 2020.

\bibitem{lin2019generalization}
S.~Lin and J.~Zhang, ``Generalization bounds for convolutional neural
  networks,'' \emph{arXiv:1910.01487}, 2019.

\bibitem{arora2018stronger}
S.~Arora, R.~Ge, B.~Neyshabur, and Y.~Zhang, ``Stronger generalization bounds
  for deep nets via a compression approach,'' \emph{ICML}, 2018.

\bibitem{bartlett2019nearly}
P.~L. Bartlett, N.~Harvey, C.~Liaw, and A.~Mehrabian, ``Nearly-tight
  {VC}-dimension and pseudodimension bounds for piecewise linear neural
  networks,'' \emph{J. Mach. Learn. Res.}, vol.~20, no.~1, pp. 2285--2301, Jan.
  2019.

\bibitem{bousquet2003new}
O.~Bousquet, ``New approaches to statistical learning theory,'' \emph{Ann.
  Inst. Statist. Math.}, vol.~55, pp. 371--389, June 2003.

\bibitem{devore2021neural}
R.~DeVore, B.~Hanin, and G.~Petrova, ``Neural network approximation,''
  \emph{Acta Numer.}, vol.~30, pp. 327--444, May 2021.

\bibitem{yarotsky2017error}
D.~Yarotsky, ``Error bounds for approximations with deep {ReLU} networks,''
  \emph{Neural Netw.}, vol.~94, pp. 103--114, Oct. 2017.

\bibitem{shen2020deep}
Z.~Shen, ``Deep network approximation characterized by number of neurons,''
  \emph{Commun. Comput. Phys.}, vol.~28, no.~5, pp. 1768--1811, June 2020.

\bibitem{hu2020deep}
Q.~Hu, F.~Gao, H.~Zhang, S.~Jin \emph{et~al.}, ``Deep learning for channel
  estimation: Interpretation, performance, and comparison,'' \emph{IEEE Trans.
  Wireless Commun.}, vol.~20, no.~4, pp. 2398--2412, Dec. 2020.

\bibitem{shrestha2023optimal}
S.~Shrestha, X.~Fu, and M.~Hong, ``Optimal solutions for joint beamforming and
  antenna selection: From branch and bound to graph neural imitation
  learning,'' \emph{IEEE Trans. Signal Process.}, vol.~71, pp. 831--846, Feb.
  2023.

\bibitem{shen2022graph}
Y.~Shen, J.~Zhang, S.~Song, and K.~Letaief, ``Graph neural networks for
  wireless communications: From theory to practice,'' \emph{IEEE Trans.
  Wireless Commun.}, vol.~22, no.~5, pp. 3554--3569, Nov. 2022.

\bibitem{schubert2019circular}
S.~Schubert, P.~Neubert, J.~P{\"o}schmann, and P.~Protzel, ``Circular
  convolutional neural networks for panoramic images and laser data,''
  \emph{IEEE IV}, 2019.

\bibitem{jo2021edge}
J.~Jo, J.~Baek, S.~Lee, D.~Kim \emph{et~al.}, ``Edge representation learning
  with hypergraphs,'' \emph{Adv. Neural Inf. Process. Syst.}, vol.~34, pp.
  7534--7546, Dec. 2021.

\bibitem{bjornson2014optimal}
E.~Bj{\"o}rnson, M.~Bengtsson, and B.~Ottersten, ``Optimal multiuser transmit
  beamforming: A difficult problem with a simple solution structure,''
  \emph{IEEE Signal Process. Mag.}, vol.~31, no.~4, pp. 142--148, July 2014.

\bibitem{trefethen1997numerical}
L.~Trefethen and D.~Bau, \emph{Numerical linear algebra}.\hskip 1em plus 0.5em
  minus 0.4em\relax Siam, June 1997, vol.~50.

\bibitem{peel2005vector}
C.~Peel, B.~Hochwald, and L.~Swindlehurst, ``A vector-perturbation technique
  for near-capacity multiantenna multiuser communication-part {I}: channel
  inversion and regularization,'' \emph{IEEE Trans. Commun.}, vol.~53, no.~1,
  pp. 195--202, Jan. 2005.

\bibitem{yoo2006optimality}
T.~Yoo and A.~Goldsmith, ``On the optimality of multiantenna broadcast
  scheduling using zero-forcing beamforming,'' \emph{IEEE J. Sel. Areas
  Commun.}, vol.~24, no.~3, pp. 528--541, Mar. 2006.

\bibitem{marzetta2010noncooperative}
T.~Marzetta, ``Noncooperative cellular wireless with unlimited numbers of base
  station antennas,'' \emph{IEEE Trans. Wireless Commun.}, vol.~9, no.~11, pp.
  3590--3600, Nov. 2010.

\bibitem{zhang2019making}
R.~Zhang, ``Making convolutional networks shift-invariant again,'' \emph{ICML},
  2019.

\bibitem{hornik1989multilayer}
K.~Hornik, M.~Stinchcombe, and H.~White, ``Multilayer feedforward networks are
  universal approximators,'' \emph{Neural Netw.}, vol.~2, no.~5, pp. 359--366,
  Jan. 1989.

\bibitem{ICC2009}
S.~Christensen, R.~Agarwal, E.~Carvalho, and J.~Cioffi, ``Weighted sum-rate
  maximization using weighted {MMSE} for {MIMO}-{BC} beamforming design,''
  \emph{IEEE ICC}, 2009.

\bibitem{3gpp38901}
\BIBentryALTinterwordspacing
3GPP, ``{Study on channel model for frequencies from 0.5 to 100 {GHz}},'' {3rd
  Generation Partnership Project (3GPP)}, Technical Report (TR) 38.901, 03
  2022, version 17.0.0. [Online]. Available:
  \url{https://portal.3gpp.org/desktopmodules/Specifications/SpecificationDetails.aspx?specificationId=3173}
\BIBentrySTDinterwordspacing

\bibitem{loyka2001channel}
S.~Loyka, ``Channel capacity of {MIMO} architecture using the exponential
  correlation matrix,'' \emph{IEEE Commun. Lett.}, vol.~5, no.~9, pp. 369--371,
  Sep. 2001.

\bibitem{araujo2019computing}
A.~Araujo, W.~Norris, and J.~Sim, ``Computing receptive fields of convolutional
  neural networks,'' \emph{Distill}, vol.~4, no.~11, p.~21, Nov. 2019.

\bibitem{hartford2018deep}
J.~Hartford, D.~Graham, K.~Brown, and S.~Ravanbakhsh, ``Deep models of
  interactions across sets,'' \emph{ICML}, 2018.

\bibitem{bartlett2017spectrally}
P.~Bartlett, D.~Foster, and M.~Telgarsky, ``Spectrally-normalized margin bounds
  for neural networks,'' \emph{NeurIPS}, 2017.

\end{thebibliography}

\appendices
\section{}\label{proof_prop_edge_gnn}

\setcounter{equation}{0}
\renewcommand\theequation{A.\arabic{equation}}

Since $\phi(\cdot)$ is an element-wise function that does not affect the permutation property, in this proof we remove $\phi(\cdot)$ for notational simplicity. Then, \eqref{eq_forward} can be rewritten as $\mathbf{D}^{(l)}=\overline{\mathbf{P}}_\mathrm{G}\odot\mathbf{D}^{(l-1)}\mathbf{1}_N\mathbf{1}^\mathsf{T}_N+\overline{\mathbf{O}}_\mathrm{G}^\prime\odot\mathbf{D}^{(l-1)}+\mathbf{1}_K\mathbf{1}^\mathsf{T}_K\overline{\mathbf{Q}}_\mathrm{G}\odot\mathbf{D}^{(l-1)}$, where $\overline{\mathbf{A}}\triangleq\left[\begin{array}{ccc}
\mathbf{A} & \cdots & \mathbf{A}\\
\vdots & \ddots & \vdots \\
\mathbf{A} & \cdots & \mathbf{A}
\end{array}\right]$, $\mathbf{A}\in\{\mathbf{O}_\mathrm{G}^\prime,\mathbf{P}_\mathrm{G},\mathbf{Q}_\mathrm{G}\}$.

Considering the structure of $\overline{\mathbf{A}}$, we have $\overline{\mathbf{A}}\odot(\boldsymbol{\Pi}_1^\mathsf{T}\mathbf{D}^{(l-1)}\boldsymbol{\Pi}_2)=\boldsymbol{\Pi}_1^\mathsf{T}(\overline{\mathbf{A}}\odot\mathbf{D}^{(l-1)})\boldsymbol{\Pi}_2$. Then, we can obtain $\mathbf{W}_\mathrm{G}(\boldsymbol{\Pi}_1^\mathsf{T}\mathbf{D}^{(l-1)}\boldsymbol{\Pi}_2)
=\overline{\mathbf{O}}_\mathrm{G}^\prime\odot(\boldsymbol{\Pi}_1^\mathsf{T}\mathbf{D}^{(l-1)}\boldsymbol{\Pi}_2)+\mathbf{1}_K\mathbf{1}_K^\mathsf{T}(\overline{\mathbf{Q}}_\mathrm{G}\odot(\boldsymbol{\Pi}_1^\mathsf{T}\mathbf{D}^{(l-1)}\boldsymbol{\Pi}_2))
+(\overline{\mathbf{P}}_\mathrm{G}\odot(\boldsymbol{\Pi}_1^\mathsf{T}\mathbf{D}^{(l-1)}\boldsymbol{\Pi}_2))\mathbf{1}_N\mathbf{1}^\mathsf{T}_N
=\boldsymbol{\Pi}_1^\mathsf{T}(\overline{\mathbf{O}}_\mathrm{G}^\prime\odot\mathbf{D}^{(l-1)})\boldsymbol{\Pi}_2+\mathbf{1}_K(\mathbf{1}_K^\mathsf{T}\boldsymbol{\Pi}_1^\mathsf{T})(\overline{\mathbf{Q}}_\mathrm{G}\odot\mathbf{D}^{(l-1)})\boldsymbol{\Pi}_2
+\boldsymbol{\Pi}_1^\mathsf{T}(\overline{\mathbf{P}}_\mathrm{G}\odot\mathbf{D}^{(l-1)})(\boldsymbol{\Pi}_2\mathbf{1}_N)\mathbf{1}^\mathsf{T}_N
\overset{\mathrm{(a)}}{=}\boldsymbol{\Pi}_1^\mathsf{T}(\overline{\mathbf{O}}_\mathrm{G}^\prime\odot\mathbf{D}^{(l-1)})\boldsymbol{\Pi}_2+(\boldsymbol{\Pi}_1^\mathsf{T}\mathbf{1}_K)\mathbf{1}_K^\mathsf{T}(\overline{\mathbf{Q}}_\mathrm{G}\odot\mathbf{D}^{(l-1)})\boldsymbol{\Pi}_2
+\boldsymbol{\Pi}_1^\mathsf{T}(\overline{\mathbf{P}}_\mathrm{G}\odot\mathbf{D}^{(l-1)})\mathbf{1}_N(\mathbf{1}^\mathsf{T}_N\boldsymbol{\Pi}_2)
=\boldsymbol{\Pi}_1^\mathsf{T}\mathbf{D}^{(l)}\boldsymbol{\Pi}_2$, where $(a)$ holds because permuting rows of $\mathbf{1}$ or columns of $\mathbf{1}^\mathsf{T}$ does not change the result.
By staking $L$ layers, we can obtained the result in the proposition according to \cite{hartford2018deep}.

\section{}\label{app_b}
We derive the estimation error bound of the EdgeGNN, while the derivations for FNN and CCNN are similar and thus not provided due to the page limitation.

According to Lemma 3.1 in \cite{bartlett2017spectrally}, the estimation error bound of a DNN can be obtained from its Rademacher complexity, which can be derived from the covering number of the weight matrix of the DNN. In what follows, we first derive the covering number of the weight matrix in each layer.

Before introducing the notion of covering number, we first define $\epsilon$-cover of a matrix set. For a set of matrices $\mathcal{W}$, its $\epsilon$-cover is another set of matrices  (denoted as $\mathrm{cover}(\mathcal{W},\epsilon)$), where for any $\mathbf{W}_1\in\mathcal{W}$, one can find $\mathbf{W}_2\in\mathrm{cover}(\mathcal{W},\epsilon)$ such that $||\mathbf{W}_1-\mathbf{W}_2||{\leq}\epsilon$, as illustrated in Fig.~\ref{fig_epsilon_cover}. For any given $\mathcal{W}$ and $\epsilon$, there may exist more than one $\epsilon$-cover of $\mathcal{W}$. The covering number (denoted as $\mathcal{N}(\mathcal{W},\epsilon)$) is the minimal number of matrices in all the $\epsilon$-covers of $\mathcal{W}$.

\begin{figure}[htb]
\centerline{\includegraphics[width=0.4\linewidth]{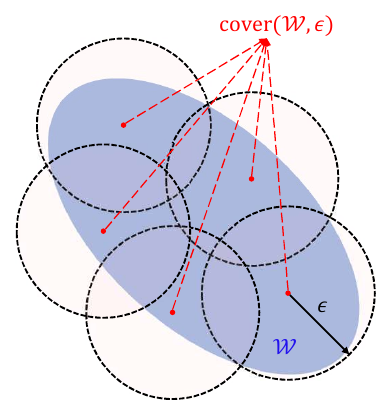}}
\bottomcaption{An $\epsilon$-cover of $\mathcal{W}$, where $\mathcal{W}$ (the blue ellipse) consists of infinite number of matrices. There are five matrices (red dots) in $\mathrm{cover}(\mathcal{W},\epsilon)$. Any matrix in $\mathcal{W}$ is within the distance $\epsilon$ from at least one matrix in $\mathrm{cover}(\mathcal{W},\epsilon)$.}
\label{fig_epsilon_cover}
\end{figure}

Denote the set of block submatrices in the $l$th layer of the EdgeGNN as $\mathcal{W}_\mathrm{S}\triangleq\{[\mathbf{O}_\mathrm{G}^{(l)\prime},\mathbf{Q}_\mathrm{G}^{(l)},
\mathbf{P}_\mathrm{G}^{(l)}]:||\mathbf{O}_\mathrm{G}^{(l)\prime}||,||\mathbf{Q}_\mathrm{G}^{(l)}||,
||\mathbf{P}_\mathrm{G}^{(l)}||{\leq}a_l/\sqrt{KN(K+N-1)}\}$, where each submatrix satisfies the F-norm bounded condition.

Denote the corresponding set of matrices  in the $l$th layer of the EdgeGNN as $\mathcal{W}_\mathrm{H}\triangleq\{\mathbf{W}_\mathrm{G}^{(l)}\mathrm{vec}(\mathbf{D}^{(l-1)}):||\mathbf{O}_\mathrm{G}^{(l)\prime}||,||\mathbf{Q}_\mathrm{G}^{(l)}||,||\mathbf{P}_\mathrm{G}^{(l)}||{\leq}a_l/\sqrt{KN(K+N-1)}\}$, which becomes the hidden representations of the $l$th layer after adding the activation function, where ${\mathbf{W}}_\mathrm{G}^{(l)}$ is the weight matrix composed of submatrices ${\mathbf{O}}_\mathrm{G}^{(l)\prime},{\mathbf{Q}}_\mathrm{G}^{(l)},{\mathbf{P}}_\mathrm{G}^{(l)}$ as in \eqref{eq_structured_parameter}.


For any $[\mathbf{O}_\mathrm{G}^{(l)\prime},\mathbf{Q}_\mathrm{G}^{(l)},\mathbf{P}_\mathrm{G}^{(l)}]\in\mathcal{W}_\mathrm{S}$, we can always find a matrix ${[\widehat{\mathbf{O}}_\mathrm{G}^{(l)\prime},\widehat{\mathbf{Q}}_\mathrm{G}^{(l)},\widehat{\mathbf{P}}_\mathrm{G}^{(l)}]}\in\mathrm{cover}(\mathcal{W}_\mathrm{S},\epsilon)$ such that $||[\widehat{\mathbf{O}}_\mathrm{G}^{(l)\prime},\widehat{\mathbf{Q}}_\mathrm{G}^{(l)},\widehat{\mathbf{P}}_\mathrm{G}^{(l)}]-[\mathbf{O}_\mathrm{G}^{(l)\prime},\mathbf{Q}_\mathrm{G}^{(l)},\mathbf{P}_\mathrm{G}^{(l)}]||{\leq}\epsilon$. Denote $\widehat{\mathbf{W}}_\mathrm{G}^{(l)}$ as the weight matrix composed of submatrices $\widehat{\mathbf{O}}_\mathrm{G}^{(l)\prime},\widehat{\mathbf{Q}}_\mathrm{G}^{(l)},\widehat{\mathbf{P}}_\mathrm{G}^{(l)}$ as in \eqref{eq_structured_parameter}. Then, according to the inequality $||\mathbf{Ab}||{\leq}||\mathbf{A}||||\mathbf{b}||$
 \cite{trefethen1997numerical}, we have $||\mathbf{W}_\mathrm{G}^{(l)}\mathrm{vec}(\mathbf{D}^{(l-1)})-\widehat{\mathbf{W}}_\mathrm{G}^{(l)}\mathrm{vec}(\mathbf{D}^{(l-1)})||{\leq}||\mathbf{W}_\mathrm{G}^{(l)}-\widehat{\mathbf{W}}_\mathrm{G}^{(l)}||||\mathrm{vec}(\mathbf{D}^{(l-1)})||=(KN||\mathbf{O}_\mathrm{G}^{(l)\prime}-\widehat{\mathbf{O}}_\mathrm{G}^{(l)\prime}||^2+KN(K-1)||\mathbf{P}_\mathrm{G}^{(l)}-\widehat{\mathbf{P}}_\mathrm{G}^{(l)}||^2+KN(N-1)||\mathbf{Q}_\mathrm{G}^{(l)}-\widehat{\mathbf{Q}}_\mathrm{G}^{(l)}||^2)^{\frac{1}{2}}||\mathrm{vec}(\mathbf{D}^{(l-1)})||{\leq}\sqrt{KN(K+N-1)}||\mathrm{vec}(\mathbf{D}^{(l-1)})||\epsilon$. Hence, $\{\widehat{\mathbf{W}}_\mathrm{G}^{(l)}\mathrm{vec}(\mathbf{D}^{(l-1)}):{[\widehat{\mathbf{O}}_\mathrm{G}^{(l)\prime},\widehat{\mathbf{Q}}_\mathrm{G}^{(l)},\widehat{\mathbf{P}}_\mathrm{G}^{(l)}]}\in\mathrm{cover}(\mathcal{W}_\mathrm{S},\epsilon)\}$ is an $\sqrt{KN(K+N-1)}||\mathrm{vec}(\mathbf{D}^{(l-1)})||\epsilon$-cover of $\mathcal{W}_\mathrm{H}$. This indicates that any $\mathrm{cover}(\mathcal{W}_\mathrm{S},\epsilon)$ can be mapped into $\mathrm{cover}(\mathcal{W}_\mathrm{H},\sqrt{KN(K+N-1)}||\mathrm{vec}(\mathbf{D}^{(l-1)})||\epsilon)$. Therefore, we can always find a $\mathrm{cover}(\mathcal{W}_\mathrm{H},\sqrt{KN(K+N-1)}||\mathrm{vec}(\mathbf{D}^{(l-1)})||\epsilon)$ with the number of matrices equal to $\mathcal{N}(\mathcal{W}_\mathrm{S},\epsilon)$, which indicates that $\mathcal{N}\left(\mathcal{W}_\mathrm{H},\sqrt{KN(K+N-1)}||\mathrm{vec}(\mathbf{D}^{(l-1)})||\epsilon\right){\leq}\mathcal{N}\left(\mathcal{W}_\mathrm{S},\epsilon\right)$.

 Since if  $||\mathbf{O}_\mathrm{G}^{(l)\prime}||,||\mathbf{Q}_\mathrm{G}^{(l)}||,||\mathbf{P}_\mathrm{G}^{(l)}||{\leq}a_l/\sqrt{KN(K+N-1)}$, then $||[\mathbf{O}_\mathrm{G}^{(l)\prime},\mathbf{Q}_\mathrm{G}^{(l)},\mathbf{P}_\mathrm{G}^{(l)}]||{\leq}a_l\sqrt{3/(KN(K+N-1))}$, $\mathcal{W}_\mathrm{S}$ is a subset of $\mathcal{W}_\mathrm{S}^\prime\triangleq\{[\mathbf{O}_\mathrm{G}^{(l)\prime},\mathbf{Q}_\mathrm{G}^{(l)},\mathbf{P}_\mathrm{G}^{(l)}]:||[\mathbf{O}_\mathrm{G}^{(l)\prime},\mathbf{Q}_\mathrm{G}^{(l)},\mathbf{P}_\mathrm{G}^{(l)}]||{\leq}\sqrt{3/(KN(K+N-1))}a_l\}$. Further considering the relation between the covering numbers of $\mathcal{W}_\mathrm{S}$ and $\mathcal{W}_\mathrm{H}$, we have
\begin{IEEEeqnarray*}{l}
\ln\mathcal{N}\left(\mathcal{W}_\mathrm{H},\epsilon\right)\\
\quad{\leq}\ln\mathcal{N}\left(\mathcal{W}_\mathrm{S},\frac{\epsilon}{\sqrt{KN(K+N-1)}||\mathrm{vec}(\mathbf{D}^{(l-1)})||}\right)\\
\quad\overset{\mathrm{(a)}}{\leq} \ln\mathcal{N}\left(\mathcal{W}_\mathrm{S}^\prime,\frac{\epsilon}{\sqrt{KN(K+N-1)}||\mathrm{vec}(\mathbf{D}^{(l-1)})||}\right)\\
\quad\overset{\mathrm{(b)}}{\leq}3C_lC_{l-1}\ln\left(1+\frac{2\sqrt{3}a_l||\mathrm{vec}(\mathbf{D}^{(l-1)})||}{\epsilon}\right),\IEEEyesnumber
\end{IEEEeqnarray*}
\noindent where (a) holds since $\mathcal{W}_\mathrm{S}$ is a subset of $\mathcal{W}_\mathrm{S}^\prime$, and (b) holds according to Lemma 10 in \cite{lin2019generalization}.

Using the similar technique as in \cite{lin2019generalization}, the covering number of all the output matrices (i.e., $\mathbf{D}^{(L)}$) obtained by inputting $\widetilde{\mathbf{H}}$ into an $L$-layer EdgeGNN can be derived.
When the first two conditions in the proposition are satisfied, the bound of the empirical Rademacher complexity of the set of functions represented by the EdgeGNN (denotes as $\mathcal{F}_\mathrm{G}$) can be derived by using similar derivations in Lemma 18 in \cite{lin2019generalization} as
\begin{equation}\label{eq_rademacher}
\mathfrak{R}_{\widetilde{\mathbf{H}}}(\mathcal{F}_\mathrm{G}){\leq}16M^{-\frac{5}{8}}\left(\frac{4L^2||\widetilde{\mathbf{H}}||\left(\prod_{l=1}^LL_\phi^{(l)}s_lR_\mathrm{G}\right)}{L_\mathrm{loss}}\right)^{\frac{1}{4}},
\end{equation}
\noindent where $R_\mathrm{G}=\sum_{l=1}^L\frac{9\sqrt{3}C_{l-1}^2C_l^2a_l}{s_l}$.

Denote the average utility of a precoding policy $f(\cdot)$ on the training set $\widetilde{\mathbf{H}}$ as $\widehat{\mathcal{R}}_{\widetilde{\mathbf{H}}}(f)=\frac{1}{M}\sum_{m=1}^MU(\mathbf{H}_m,f(\mathbf{H}_m))$. Then, the estimation error can be bounded as
\begin{IEEEeqnarray*}{rcl}\label{eq_UB}
\mathcal{R}(\hat{f}^\star)-\mathcal{R}(\hat{f})&=&
\mathcal{R}(\hat{f}^\star)-\mathcal{R}(\hat{f})+
\widehat{\mathcal{R}}_{\widetilde{\mathbf{H}}}(\hat{f}^\star)-
\widehat{R}_{\widetilde{\mathbf{H}}}(\hat{f}^\star)\\
&{\leq}&\mathcal{R}(\hat{f}^\star)-\mathcal{R}(\hat{f})+
\widehat{\mathcal{R}}_{\widetilde{\mathbf{H}}}(\hat{f})-
\widehat{R}_{\widetilde{\mathbf{H}}}(\hat{f}^\star)\\
&{\leq}&|\mathcal{R}(\hat{f}^\star)
-\widehat{R}_{\widetilde{\mathbf{H}}}(\hat{f}^\star)|+
|\widehat{\mathcal{R}}_{\widetilde{\mathbf{H}}}(\hat{f}^\star)-
\mathcal{R}(\hat{f})|\\
&{\leq}&2\sup_{f\in\mathcal{F}}|\mathcal{R}(f)-
\widehat{\mathcal{R}}_{\widetilde{\mathbf{H}}}(f)|.\IEEEyesnumber
\end{IEEEeqnarray*}
According to Lemma 3.1 in \cite{bartlett2017spectrally},
$\sup_{f\in\mathcal{F}}|\mathcal{R}(f)-\widehat{\mathcal{R}}_{\widetilde{\mathbf{H}}}(f)|{\leq}2\mathfrak{R}_{\widetilde{\mathbf{H}}}(\mathcal{F}_\mathrm{G})+3\sqrt{\ln(1/\delta)/2M}$ with probability $1-\delta$. 

Upon substituting \eqref{eq_rademacher}, we complete the proof.

\vfill

\end{document}